\documentclass[11pt]{article}
\usepackage{graphicx,amsmath,amsfonts,amssymb,amsthm,bbold,bm, bbm,dsfont,authblk,mathtools,fullpage,braket,tikz,array,float}
\usepackage{mathrsfs}
\usepackage[dvipsnames,svgnames]{xcolor}
\usepackage[colorlinks]{hyperref}
\usepackage[capitalize]{cleveref}
\usepackage[title]{appendix}
\usepackage[labelfont=bf]{caption}
\usepackage{multirow}
\usepackage{thmtools} 
\usepackage{subcaption}
\usepackage{booktabs}

\usepackage[backend=biber, style=alphabetic]{biblatex}
\crefformat{equation}{(#2#1#3)}
\Crefformat{equation}{(#2#1#3)}
\crefname{section}{\S}{\S\S}
\Crefname{section}{\S}{\S\S}
\crefformat{section}{\S#2#1#3}
\Crefformat{section}{\S#2#1#3}

\DeclarePairedDelimiter{\abs}{\lvert}{\rvert}

\DeclareMathOperator*{\argmin}{arg\,min}

\DeclareMathOperator{\diag}{diag}
\DeclareMathOperator{\im}{im}
\DeclareMathOperator{\supp}{supp}
\DeclareMathOperator{\Tr}{Tr}
\DeclareMathOperator{\Var}{Var}

\newcommand{\E}{\mathbb E}
\newcommand{\PS}{P_{\mathrm S}^{\mathrm{ML}}}
\newcommand{\dc}{\delta_\mathrm{cyc}}
\newcommand{\wt}[1]{\lvert #1\rvert}
\newcommand{\tr}{\mathrm{tr}}

\newtheorem{definition}{Definition}
\newtheorem{proposition}{Proposition}
\newtheorem{theorem}{Theorem}
\newtheorem{corollary}{Corollary}
\newtheorem{lemma}{Lemma}
\newtheorem{remark}{Remark}
\newtheorem{conjecture}{Conjecture}
\crefname{conjecture}{Conjecture}{Conjectures}

\title{Cycle Codes and Decoded Quantum Interferometry}

\makeatletter
\renewcommand{\AB@authnote}[1]{{\textsuperscript{\normalfont#1}}}
\renewcommand{\AB@affilnote}[1]{{\textsuperscript{\normalfont#1}}}
\makeatother

\author[1]{Anuj Apte}
\author[1]{Shouvanik Chakrabarti}
\author[2]{Andi Gu}
\author[3]{Stephen P. Jordan}
\author[4]{Ojas Parekh}
\author[1]{Ruslan~Shaydulin}
\author[1]{Jacob Watkins\thanks{Correspondence: \href{mailto:jacob.a.watkins@jpmchase.com}{jacob.a.watkins@jpmchase.com}}}
\author[3]{Noureldin Yosri}
\author[3]{Adam Zalcman}

\affil[1]{Global Technology Applied Research, JPMorganChase, New York, NY}
\affil[2]{Harvard University, Cambridge, MA}
\affil[3]{Google Quantum AI, Venice, CA}
\affil[4]{Sandia National Laboratories, Albuquerque, NM}

\date{}

\begin{document}

\maketitle

\begin{abstract}
Decoded Quantum Interferometry (DQI) reduces optimization problems with two-variable constraints to decoding cycle codes. For one such problem, namely MaxCut, prior work showed that DQI achieves a nontrivial satisfaction fraction guarantee only on linear-girth graphs, for which MaxCut is classically easy. However, these no-go results rely on minimum distance assumptions that do not represent true decodability thresholds for common noise channels, thus underestimating actual DQI performance. To estimate the true performance, we derive DQI satisfaction guarantees in the presence of imperfect decoding for fixed instances and generalize prior results for random instances. We then study homological cycle codes over arbitrary finite fields and show that minimum-weight decoding is NP-hard for every field of size $q>2$, in contrast with known results for binary cycle codes. On the Linial--Simkin ensemble of regular graphs with logarithmic girth, we prove maximum-likelihood recovery bounds and a strong converse for a family of non-uniform additive noise channels. Moreover, linear programming (LP) decoders provide polynomial-time recovery of a positive fraction of randomly located errors with arbitrary values. These results allow us to prove upper and lower bounds on the approximate optima achievable by DQI when restricted to classical decoders. These bounds rule out quantum advantage in the regimes we analyze but also identify a family of regular Max-$k$-Cut instances on which DQI efficiently achieves a nontrivial satisfaction fraction guarantee. 
\end{abstract}

\clearpage

\tableofcontents

\clearpage

\section{Introduction}

In many commonly studied optimization problems, one is given a list of constraints, each of which involves only a small subset of the decision variables, and seeks to maximize the number of constraints satisfied. Examples include Max-$3$-SAT and MaxCut, both of which are NP-hard. Decoded Quantum Interferometry (DQI) is a quantum algorithm that uses a quantum Fourier transform to reduce such optimization problems to decoding problems for classical linear error correction codes \cite{JSW25}. The objective value of the approximate solution is determined by the number of symbol-flip errors that can be decoded. The existence of efficient decoding algorithms for low-density parity-check (LDPC) codes at high error rates raises the tantalizing possibility of quantum advantage on the corresponding optimization problems. However, the competing classical optimization heuristics are quite powerful, and such quantum advantage has not yet been realized for these sparse families.

The specific family of discrete optimization problems to which DQI has been applied is known as Max-LINSAT.
\begin{definition}[Max-LINSAT]\label{def:maxlinsat}
    Let $B\in\mathbb F_q^{m\times n}$ and let $C_i\subseteq\mathbb F_q$ for $i =1,\ldots, m$. Max-LINSAT asks for $\mathbf x\in\mathbb F_q^n$ satisfying as many as possible of the constraints $\sum_jB_{ij}x_j\in C_i$.
\end{definition}
\noindent Under the DQI reduction, $B^T$ becomes a parity-check matrix associated with a syndrome decoding problem. The clearest examples of quantum advantage via DQI arise for Max-LINSAT instances in which $\ker B^T$ is an algebraic code \cite{JSW25, gu2025algebraic}. In \cite{CLZ22}, a quantum Fourier transform is used to reduce a constraint satisfaction problem to a decoding problem for codes defined by a dense, algebraically unstructured parity-check matrix. A quantum decoding algorithm is then used to solve the decoding problem, yielding a potential quantum advantage. Subsequent work \cite{KOW25} gave a classical counterattack eliminating quantum advantage in many parameter regimes, but whether the algorithm of \cite{CLZ22} can retain advantage in some domain is an unresolved question.

Here we investigate the application of DQI to Max-LINSAT problems in which each constraint contains exactly two variables, i.e., Max-2-LINSAT. Over $\mathbb F_q$, the constraints $x_v-x_u\ne0$ give Max-$k$-Cut for prime powers $k = q$, while $q=2$ gives MaxCut. Under DQI, such problems are reduced to decoding codes in which each symbol is contained in exactly two parity checks. Codes over a binary alphabet with this property are known as \emph{cycle codes}, and they behave very differently from codes whose variables participate in at least three parity checks, i.e., conventional LDPC codes. Algorithmic and complexity-theoretic aspects of DQI for constraints involving three or more variables have seen substantial progress over both binary and nonbinary alphabets~\cite{kramer2026approximability, marwaha2025complexity}. By contrast, prior results for two-variable constraints are limited to two negative results concerning quantum advantage with classical decoders in the binary setting~\cite{P25, JSW25}.

Cycle codes are directly associated with graphs, where each edge represents a symbol in the codeword, and the vertices at its two ends are the two parity checks in which it is contained. For regular graph families of fixed degree at least three, cycle codes have minimum distance at most logarithmic in block size. Thus, one cannot guarantee recovery from all error patterns with more than logarithmically many symbol errors. The number of errors introduced, and successfully corrected, in the DQI algorithm corresponds to the quality of the approximate solution to the Max-LINSAT instance. If only logarithmically many errors are introduced, DQI does not perform asymptotically better than random guessing. Furthermore, for graph families whose cycle codes have linear minimum distance, it has been shown that the DQI-associated MaxCut problem is exactly solvable in classical polynomial time~\cite{P25}. 

Despite their low minimum distance, suitable ensembles of regular graphs permit reliable decoding of a linear number of random errors~\cite{DZ97,NVZ16}. This contrasts sharply with the logarithmic number of errors covered by worst-case guarantees. In classical error correction, this distinction is captured by noise thresholds for a specified communication channel, below which decoding succeeds with high probability. For DQI, the corresponding question is whether reliable decoding under the error distribution induced by the algorithm yields stronger optimization guarantees. Answering this requires an analysis that allows some syndromes in the DQI state to be decoded incorrectly. Prior work provides such an analysis when the target sets $C_i$ are randomly distributed. For problems such as Max-$k$-Cut, however, the target sets are fixed, and those guarantees do not apply.

This work fills this gap by proving a version of the semicircle law that applies to fixed Max-LINSAT target sets and accounts for decoding errors. Combining our fixed-target semicircle law (\cref{thm:fixed-target-semicircle}) with our linear programming (LP) decoding threshold bound (\cref{thm:lp-reliability,cor:lp-linial}), we show that DQI achieves, in polynomial time, satisfaction fractions for MaxCut and Max-$k$-Cut that remain asymptotically above the random-guessing baseline. However, these gains remain insufficient to achieve quantum advantage over classical algorithms with rigorous guarantees, such as the classical vector algorithm of Thompson, Parekh, and Marwaha (TPM) and its recent generalization to Max-$k$-Cut~\cite{Frieze1997, TPM22, apte2026provable}. Our algorithmic conclusions are summarized in \Cref{tab:DQI-MaxkCut-summary}.

\begin{table}[htbp] 
    \centering
    \small
    \setlength{\tabcolsep}{5pt}
    \renewcommand{\arraystretch}{1.2}
    \begin{tabular}{@{}llcl@{}}
    \toprule
    \textbf{Problem} & \textbf{Algorithm} & \textbf{Worst-case runtime} & \textbf{Expected cut fraction} \\
    \midrule
    \multirow{3}{*}{MaxCut}
        & Random Guess   & $O(n)$ & $\frac12$ \\
        & TPM      & $\operatorname{Poly}(n)$ & $\frac12+\frac{2}{\pi\sqrt D}+O(D^{-1})$ \\
        & DQI + ML & $\operatorname{Poly}(n)$ & $\frac12\left(1+\frac1{D-1}\right)$ \\
    \midrule
    \multirow{4}{*}{Max-$k$-Cut ($k>2$)}
        & Random Guess  & $O(n)$ & $\frac{k-1}{k}$ \\
        & TPM      & $\operatorname{Poly}(n)$ & $\frac{k-1}{k}+\frac{\kappa_k}{\sqrt D}+O_k(D^{-1})$ \\
        & DQI + ML & $\operatorname{Exp}(n)$  & $\frac{k-1}{k}\left(1+\frac1{D-1}\right)+o(D^{-1})$ \\
        & DQI + LP & $\operatorname{Poly}(n)$ & $\frac{k-1}{k}\left(1+\frac1{\sqrt{k-1}(D-1)^2}\right)+o(D^{-2})$ \\
    \bottomrule
    \end{tabular}
    \caption{Algorithmic comparison for MaxCut and Max-$k$-Cut, where $k>2$ is a fixed prime power. The cut fractions concern the Linial--Simkin graph family of \cref{def:linial-simkin-ensemble}, with the large-graph limit taken before the large-$D$ limit. The constant $\kappa_k$ is defined in \cref{eq:tpm-large-degree}; the DQI + LP entry uses the $c\uparrow1$ limit of the certified rate in \cref{eq:lp-certified-rate}. See \cref{sec:max_k_cut} for details.}
    \label{tab:DQI-MaxkCut-summary}
\end{table}

To analyze DQI performance on Max-$k$-Cut, one must investigate a generalization of cycle codes that permits symbols over general finite fields. The behavior of such codes was poorly understood prior to the present work. One of the main contributions of this paper is to initiate the study of $q$-ary cycle codes (hereafter simply ``cycle codes''). We resolve several foundational questions about the decoding complexity and reliability of \emph{homological} cycle codes, in which the two nonzero coefficients in a column of the parity-check matrix sum to zero. These results are summarized in \cref{tab:binary_codes}.

\begin{table}[!htbp]
\centering
\small
\setlength{\tabcolsep}{6pt}
\renewcommand{\arraystretch}{1.04}
\captionsetup{font=small,skip=6pt}
\begin{tabular}{@{}>{\raggedright\arraybackslash}p{.205\textwidth}
                   >{\raggedright\arraybackslash}p{.325\textwidth}
                   >{\raggedright\arraybackslash}p{\dimexpr.47\textwidth-4\tabcolsep\relax}@{}}
\toprule
 & \textbf{Binary cycle codes} & \textbf{Nonbinary cycle codes (this paper)}\\
 & $\mathbb F_2$ & $\mathbb F_q$, any prime power $q>2$\\
\midrule
parity-check structure
 & $2$ checks per bit;\newline
   $D$ bits per check~\cite{NVZ16}
 & $2$ checks per symbol;\newline
   $D$ symbols per check (\cref{sec:decoding-cycle-codes})\\
\midrule
Code rate
 & $\displaystyle 1-\frac2D+\frac1m$~\cite{NVZ16}
 & $\displaystyle 1-\frac2D+\frac1m$ (\cref{sec:decoding-cycle-codes})\\
\midrule
Minimum distance
 & $d_{\min}=g_G$,\newline
   $g_G \le 2\log_{D-1}n+2$~\cite{NVZ16}
 & $d_{\min} = g_G$ (same as $q=2$)\\
\midrule
Guaranteed unique-decoding radius
 & $\displaystyle t_{\mathrm{wc}}=\left\lfloor\frac{g_G-1}{2}\right\rfloor$~\cite{NVZ16}
 & $\displaystyle t_{\mathrm{wc}}=\left\lfloor\frac{g_G-1}{2}\right\rfloor$ (\cref{sec:decoding-cycle-codes})\\
\midrule
ML reliability threshold:\newline
 & $\displaystyle \delta_\mathrm{rel}(D,2)=\delta_\mathrm{cyc}(D,2)$, \newline where $\delta_\mathrm{cyc}$ is as defined in \cref{eq:ml-cycle-value}. Proven in~\cite{DZ97,NVZ16}.
 & If $q/D^{2/3}\le1/32$ then \newline 
 $\displaystyle \left(1-\frac{16 q}{D^{2/3}} \right) \delta_\mathrm{cyc} \leq \delta_\mathrm{rel}(D,q) \leq \delta_\mathrm{cyc}(D,q),$\newline
   where $\delta_\mathrm{cyc}(D,q)$ is as defined in \cref{eq:ml-cycle-value}. Proven in \cref{sec:ml-threshold} and \cref{app:ml-small-alphabet}.\\
\midrule
Exact ML / nearest-codeword computation
 & Minimum-weight $T$-join; polynomial time~\cite{EJ73}
 & NP-hard for every fixed $q>2$\newline
   (\cref{sec:ml-hardness}, \cref{thm:hardness})\\
\midrule
Efficient average-case decoding
 & The same $\delta_\mathrm{rel}(D,2)$ is attained by exact ML decoding~\cite{EJ73}.
 & Polynomial-time LP decoding corrects $\Theta(m)$ randomly located symbol errors: $P_{\mathrm W}^{\mathrm{LP}}\to0$ for every $0<p<p_{\mathrm{LP}}^{\mathrm{cert}}(D,c)$ (\cref{cor:lp-linial}).\newline
   This certified rate need not equal the ML threshold. Observed recovery extends beyond this rate (\cref{sec:lp-experiments}).\\
\bottomrule
\end{tabular}
\caption{\label{tab:binary_codes}\textbf{Binary versus nonbinary homological cycle codes.} Here $m$ is the block length, $n$ is the number of parity checks, $D$ is the graph degree (therefore $m=Dn/2$), $g_G$ is the graph girth, and $\delta_{\mathrm{rel}}(D, q)$ is the supremum of noise rates permitting reliable ML recovery. For independent uniformly distributed nonzero errors, maximum-likelihood (ML) decoding limits concern whole-word error averaged over the Linial--Simkin ensemble of \cref{def:linial-simkin-ensemble}, as $n\to\infty$ at fixed degree $D$, fixed finite field size $q$, fixed Linial--Simkin girth parameter $c$, and fixed error probability $p\in(0,1-1/q)$.}
\end{table}

Let us summarize these code-theoretic results. Suppose each symbol is independently replaced with a uniformly chosen different symbol with probability $p$. We call decoding \emph{reliable} if the probability of recovering the entire original codeword tends to one as the block length grows. Uniformly random $D$-regular graphs, with fixed $D\ge3$, do not permit reliable decoding at any fixed noise rate $0<p<1-1/q$: short cycles occur with nonvanishing probability and create decoding ambiguities that persist even for maximum-likelihood (ML) decoding. \cref{thm:unexpurgated-error-floor} establishes this for every finite field.

To remove this obstruction, we study the \emph{Linial--Simkin random greedy ensemble}~\cite{LS}, which produces regular graphs whose girth grows logarithmically with the number of vertices. This ensemble excludes short cycles while retaining enough randomness to analyze decoding performance. Its construction is given in \cref{def:linial-simkin-ensemble}. All recovery probabilities below are averaged over both the random graph and the random symbol errors.

A central quantity in our analysis is the \emph{simple-cycle value} $\delta_{\mathrm{cyc}}(D,q)$. A competing codeword can differ from the transmitted codeword along a single cycle. The simple-cycle value balances the abundance of these competitors against the probability that noise makes them at least as likely as the transmitted codeword. We prove that above this value, ML decoding fails with probability tending to one, for every fixed degree $D\ge3$ and field size $q$. Thus, simple cycles alone impose an upper limit on reliable recovery.

Our recovery bounds approach this limit when $q=o(D^{2/3})$, including any fixed alphabet as $D$ grows. In this regime, both bounds have leading-order value $(q-1)/(4D^2)$. More explicitly, for any fixed $\varepsilon\in(0,1)$ and sufficiently large $D$,
\begin{equation}\label{eq:cyc2}
    \Pr(\text{ML decoding recovers the entire codeword})
    \longrightarrow
    \begin{cases}
        1, & 0<p<(1-\varepsilon)\dfrac{q-1}{4D^2},\\[5pt]
        0, & (1+\varepsilon)\dfrac{q-1}{4D^2}<p<1-\dfrac1q.
    \end{cases}
\end{equation}
Here the block length tends to infinity first, with the degree, alphabet, and ensemble parameters fixed. The explicit bounds, including their finite-degree gap, appear in \cref{sec:ml-threshold}. For binary cycle codes, the reliability threshold equals the simple-cycle value at every degree~\cite{DZ97,NVZ16}; our nonbinary bounds establish asymptotic agreement in the regime above.

These optimal-decoding guarantees do not by themselves provide efficient algorithms. Indeed, we prove that minimum-weight decoding is NP-hard for every fixed $q>2$, unlike the binary case. Nevertheless, our polynomial-time LP decoder reliably corrects a positive fraction of randomly located errors, uniformly over their nonzero values. Its certified error rate need not attain the ML threshold. \cref{sec:ml-hardness} and \cref{sec:lp-decoder} establish the complexity and LP recovery results, respectively.

\paragraph{Notation and conventions} Throughout this manuscript, $q$ is a prime power and $\mathbb F_q$ is the unique (up to isomorphism) finite field of size $q$. As a caution, note $\mathbb{F}_q \neq \mathbb{Z}_q$ when $q$ is not prime. We write $\mathbb F_q^\times=\mathbb F_q\setminus\{0\}$ for the multiplicative group. For a vector $\mathbf v \in \mathbb{F}^n$ over a field $\mathbb{F}$, let $\supp \mathbf v$ be the set of indices of its nonzero entries, and $\abs{\mathbf v} \coloneqq \abs{\supp \mathbf v}$ the Hamming weight. The Hamming distance between words $\mathbf x$ and $\mathbf y$ is $\abs{\mathbf x - \mathbf y}$. For a graph $G$, write $g_G$ for its girth, $\delta_G(u,v)$ for the distance between vertices $u$ and $v$, and $d_G(u)$ for the degree of vertex $u$. Adopt $g_G = \infty$ for acyclic $G$. The indicator function for a set $A$ is $\mathbb 1_A$.

\section{Decoded Quantum Interferometry with imperfect decoding} \label{sec:dqi-imperfect-decoding}

Here we analyze how imperfect syndrome decoding affects the assignments produced by DQI. We first describe the prepared error state and the postselected output, then relate the expected satisfaction fraction to decoding performance for fixed targets. We subsequently show how averaging over random targets permits a weaker decoding hypothesis.

Let $B\in\mathbb F_q^{m\times n}$ have rows $\mathbf b_i$, and fix target sets $\mathbf C=(C_1,\ldots,C_m)$ with $C_i\subseteq\mathbb F_q$ and $|C_i|=r$. We assume $0 < r < q$, since the cases $r = 0$ and $r = q$ correspond to trivial optimization problems. The Max-LINSAT objective is
\begin{equation} \label{eq:dqi-satisfaction-function}
    s(\mathbf x) \coloneqq \sum_{i=1}^m\mathbb 1_{C_i}(\mathbf b_i\cdot\mathbf x), \qquad \mathbf x\in\mathbb F_q^n.
\end{equation}
Let $\rho=r/q$, and define the centered, normalized target indicators
\begin{equation} \label{eq:dqi-centered-functions}
    g_i(x)\coloneqq \frac{\mathbb 1_{C_i}(x)-\rho}{\sqrt{q\rho(1-\rho)}}.
\end{equation}
Centering removes the zero Fourier component, while normalizing ensures $\abs{g_i}^2$ is a probability distribution. For $q=\ell_0^\nu$ with $\ell_0$ prime, define the additive character
\begin{equation}\label{eq:field-character}
    \chi_q(t)\coloneqq\exp\!\left(\frac{2\pi i}{\ell_0}
    \Tr_{\mathbb F_q/\mathbb F_{\ell_0}}(t)\right),
    \qquad
    \Tr_{\mathbb F_q/\mathbb F_{\ell_0}}(t)
    \coloneqq\sum_{j=0}^{\nu-1}t^{\ell_0^j}.
\end{equation}
The character function satisfies $\chi_q(s+t)=\chi_q(s)\chi_q(t)$, $\overline{\chi_q(t)}=\chi_q(-t)$, and
\begin{equation} \label{eq:finite-field-character-orthogonality}
    \sum_{\mathbf x\in\mathbb F_q^d}
    \chi_q(\mathbf a\cdot\mathbf x)
    =q^d\mathbbm 1_{\mathbf 0}(\mathbf a).
\end{equation}
Thus, the Fourier transform with entries $(\mathsf F_q)_{a,x}=\chi_q(ax)/\sqrt q$ is unitary. Let $\widetilde{g}_i$ be the unitary Fourier transform of $g_i$. By Parseval's identity, $\sum_{a\neq 0} \abs{\widetilde{g}_i(a)}^2 = 1$, so these squared magnitudes define a distribution over nonzero error values. For an error word $\mathbf y\in\mathbb F_q^m$, define
\begin{equation} \label{eq:dqi-product-fourier-coefficient}
    \widetilde g_{\mathbf C}(\mathbf y)\coloneqq
    \prod_{i:y_i\ne0}\widetilde g_i(y_i),
\end{equation}
with the empty product equal to one.

The $k$th \emph{Hamming shell} consists of the errors of weight $k$. DQI begins by preparing an error state specified by $\widetilde{g}_\mathbf{C}$ and a normalized amplitude vector $\mathbf w=(w_0,\ldots,w_m)\in\mathbb C^{m+1}$. 
\begin{equation} \label{eq:dqi-full-error-state}
    \ket\Gamma\coloneqq
    \sum_{\mathbf y\in\mathbb F_q^m}\phi_{\mathbf C}(\mathbf y)\ket{\mathbf y},
    \qquad
    \phi_{\mathbf C}(\mathbf y)\coloneqq
    \frac{w_{|\mathbf y|}}{\sqrt{\binom m{|\mathbf y|}}}
    \widetilde g_{\mathbf C}(\mathbf y)
\end{equation}
Whereas $\mathbf w$ controls the weight distribution over the Hamming shells, $\widetilde{g}_\mathbf{C}$ determines the amplitudes within each shell. More precisely, since
\begin{equation} \label{eq:dqi-shell-normalization}
    \sum_{|\mathbf y|=k}|\widetilde g_{\mathbf C}(\mathbf y)|^2 = \binom mk,
\end{equation}
this state is normalized and assigns probability $|w_k|^2$ to errors of weight $k$. For state preparation, let $\ket{D_k^m} \propto \sum_{\mathbf z\in\{0,1\}^m:|\mathbf z|=k}\ket{\mathbf z}$ be the normalized Dicke state. Applying $U=\bigotimes_iU_i$ to $\ket{\mathbf w}\coloneqq\sum_kw_k\ket{D_k^m}$ gives $\ket\Gamma$, where $U_i$ are the isometries into $\mathbb{C}^q$ given by
\begin{equation} \label{eq:Ui_def}
    U_i\ket0=\ket0,
    \qquad
    U_i\ket1=\sum_{a\in\mathbb F_q^\times}\widetilde g_i(a)\ket a.
\end{equation}

The decoder acts on syndromes of $C^\perp=\ker B^T$. We assume that $D:\operatorname{im}B^T\rightarrow\mathbb F_q^m$ is deterministic and syndrome-consistent:
\begin{equation} \label{eq:syndrome-consistent-decoder}
    B^TD(\mathbf s)=\mathbf s
    \qquad\text{for all }\mathbf s\in\operatorname{im}B^T.
\end{equation}
Syndrome consistency requires the decoder to return an error compatible with its input syndrome, not necessarily the original error prepared. Computing the syndrome produces
\begin{equation} \label{eq:dqi-before-decoding}
    \sum_{\mathbf y}\phi_{\mathbf C}(\mathbf y)
    \ket{\mathbf y}\ket{B^T\mathbf y}.
\end{equation}
A reversible implementation of $D$ then subtracts $D(B^T\mathbf y)$ from the error register. Postselecting this register on $\ket{\mathbf0}$ retains exactly the decoder's success set
\begin{equation} \label{eq:decoder-success-failure-partition}
    \mathcal S \coloneqq
    \{\mathbf y:D(B^T\mathbf y)=\mathbf y\}.
\end{equation}
Conditional on successful postselection, the resulting syndrome state is
\begin{equation} \label{eq:dqi-postselected-syndrome-state}
    \ket{\widetilde P_{\mathcal S}}
    \propto\sum_{\mathbf y\in\mathcal S}
    \phi_{\mathbf C}(\mathbf y)\ket{B^T\mathbf y}.
\end{equation}
The algorithm applies $(\mathsf F_q^\dagger)^{\otimes n}$ and measures in the computational basis to obtain a random assignment $\mathbf X$. All expected satisfaction fractions below are conditioned on successful postselection. The analysis applies to any implementation of the state preparation and reversible decoder described above.

Define, for $0\leq p \leq 1$,
\begin{equation}\label{eq:fixed-target-semicircle-function}
    F_\rho(p) \coloneqq 
    \begin{cases}
        \left(\sqrt{\rho(1-p)}+\sqrt{(1-\rho)p}\right)^2, & p \in [0,1-\rho] \\
        1, & p \in (1-\rho, 1].
    \end{cases}
\end{equation}
The function $F_\rho$ describes the asymptotic ideal satisfaction fraction optimized over amplitudes supported on errors of weight at most $pm$. It saturates at one because placing amplitude on shells beyond $(1-\rho)m$ is not required to achieve perfect satisfaction. Let $P_k$ be the error distribution obtained by choosing $k$ coordinates uniformly and, at each chosen coordinate $i$, independently drawing a nonzero value $a$ with probability $\abs{\widetilde{g}_i(a)}^2$. This is the DQI error distribution conditioned on weight $k$.

\paragraph{Satisfaction and one-coordinate shifts.}
To motivate the decoding hypothesis below, we express the satisfaction observable in the Fourier basis. Define
\begin{equation}\label{eq:satisfaction-objective}
    \mathsf S\coloneqq\sum_{\mathbf x}s(\mathbf x)
        \ket{\mathbf x}\!\bra{\mathbf x},
    \qquad
    \mathsf S_\mathrm{err}\coloneqq\sum_{i=1}^m\Pi_i,
    \qquad
    \Pi_i\coloneqq\sum_{t\in C_i}\ket t\!\bra t,
\end{equation}
where $\Pi_i$ acts on coordinate $i$ of the error register. The \emph{auxiliary observable} $\mathsf S_\mathrm{err}$ tests the error coordinates directly, rather than the linear forms $\mathbf b_i\cdot\mathbf x$. Write $Z_{\mathbf v}\ket{\mathbf x}=\chi_q(\mathbf v\cdot\mathbf x)\ket{\mathbf x}$ and $X_{\mathbf v}\ket{\mathbf x}=\ket{\mathbf x+\mathbf v}$ for the $q$-ary Pauli operators. The unitary Fourier coefficients of the target indicators are
\begin{equation}\label{eq:target-indicator-fourier-coefficients}
    \widetilde{\mathbb1}_{C_i}(a)\coloneqq
        \frac1{\sqrt q}\sum_{t\in C_i}\chi_q(at).
\end{equation}
Fourier inversion gives the Pauli expansion
\begin{equation}\label{eq:satisfaction-pauli-expansion}
    \mathsf S=\frac1{\sqrt q}\sum_{i,a}
        \widetilde{\mathbb1}_{C_i}(-a)Z_{a\mathbf b_i}.
\end{equation}
Here and below, $i\in[m]$ and $a\in\mathbb F_q$ unless otherwise indicated. For an operator $A$ on $d$ coordinates, write $\widetilde A\coloneqq\mathsf F_q^{\otimes d}A(\mathsf F_q^\dagger)^{\otimes d}$. Since Fourier conjugation sends $Z_{\mathbf v}$ to $X_{-\mathbf v}$, we obtain
\begin{align}
    \widetilde{\mathsf S} &=\frac1{\sqrt q}\sum_{i,a} \widetilde{\mathbb1}_{C_i}(-a)X_{-a\mathbf b_i}, \label{eq:satisfaction-fourier-shifts}\\
    \widetilde{\mathsf S}_\mathrm{err} &=\frac1{\sqrt q}\sum_{i,a} \widetilde{\mathbb1}_{C_i}(-a)X_{-a\mathbf e_i}, \label{eq:fourier-Serr-X}
\end{align}
where $\mathbf e_i$ is the $i$th standard basis vector of $\mathbb F_q^m$. Each syndrome shift in $\widetilde{\mathsf S}$ can thus be realized by changing at most one error coordinate, since $B^T(\mathbf y-a\mathbf e_i)=B^T\mathbf y-a\mathbf b_i$. If the shifted error is not successfully decoded, the decoder instead selects a different error word with the same syndrome, giving an additional interference contribution.

Our fixed-target guarantee therefore assumes more than recovery of a typical error: we require recovery of all errors differing from it in at most one coordinate. This condition controls the additional interference just described.
\begin{theorem}[Semicircle law for fixed targets]
\label{thm:fixed-target-semicircle}
\label{thm:DQI-error-main}
    Let $\ell \leq m$ be a nonnegative integer, and let $\varepsilon \in [0,1)$. Suppose that, for every $0\le k\le\ell$,
    \begin{equation}\label{eq:robust-fixed-weight-decoding-hypothesis}
        \Pr_{\mathbf Y\sim P_k}\!\left[
            D(B^T\mathbf Y')=\mathbf Y'
            \text{ for every $\mathbf Y'$ with }|\mathbf Y - \mathbf Y'|\le1
        \right]\ge1-\varepsilon.
    \end{equation}
    Let $f_\ell^*$ be the largest value of $\mathbb E[s(\mathbf X)/m\mid\text{successful postselection}]$ over all normalized shell-amplitude vectors $\mathbf w=(w_0,\ldots,w_m)^T$ satisfying $w_k=0$ for $k>\ell$. Then
    \begin{equation}\label{eq:fixed-target-optimized-semicircle}
    \begin{aligned}
        f_\ell^* &\ge (1-\varepsilon)\left(F_\rho(\ell/m)-\frac1{\sqrt m}\right) - \frac{\sqrt{(q-1)\rho(1-\rho)}\,\varepsilon}{1-\varepsilon}, \\
        f_\ell^* &\le (1-\varepsilon)F_\rho(\ell/m)+\varepsilon + \frac{\sqrt{(q-1)\rho(1-\rho)}\,\varepsilon}{1-\varepsilon}.
    \end{aligned}
    \end{equation}
    Every such normalized shell-amplitude vector has postselection probability at least $1-\varepsilon$.
\end{theorem}
\noindent For fixed $q$ and $\rho$, a vanishing one-coordinate recovery failure probability preserves the ideal semicircle prediction, up to a finite-size correction. In particular, if $\ell/m\rightarrow p$ and $\varepsilon \rightarrow 0$ as $m \rightarrow \infty$, the optimized satisfaction fraction approaches $F_\rho(p)$, while the postselection success probability approaches one.
\begin{proof}
    We first restrict the auxiliary satisfaction observable to the shell-amplitude subspace. We then separate the effects of postselection and additional syndrome-space interference, control these effects using the recovery hypothesis, and finally estimate the largest eigenvalue of the restricted shell matrix.
    
    To describe the auxiliary observable in shell coordinates, define the isometric embedding
    \begin{equation}
        \mathsf J: \mathbb{C}^{m+1} \rightarrow \left(\mathbb{C}^q\right)^{\otimes m}, \qquad \mathsf J \mathbf w \coloneqq U \sum_{k=0}^m w_k \ket{D_k^m}.
    \end{equation}
    Thus, $\mathsf J \mathbf w = \ket{\Gamma}$, and its image is the subspace spanned by the prepared shell states. 

    To determine the observable's action on this subspace, first consider a single coordinate. The span of the uniform vector $\ket{u} = q^{-1/2} \sum_x \ket{x}$ and the centered indicator $\ket{g_i} = \sum_x g_i(x) \ket{x}$ is invariant under $\Pi_i$, because both vectors are constant on $C_i$ and on its complement. Since $\mathsf F_q^\dagger U_i$ maps $\ket{0}$ and $\ket{1}$ to $\ket{u}$ and $\ket{g_i}$, respectively, this invariant subspace gives
    \begin{equation}\label{eq:Pi-tilde-basis}
        U_i^\dagger \widetilde\Pi_i U_i =
        \begin{bmatrix}
            \rho & \sqrt{\rho(1-\rho)}\\
            \sqrt{\rho(1-\rho)} & 1-\rho
        \end{bmatrix}.
    \end{equation}
    The sum of these identical one-qubit operators preserves the symmetric qubit subspace. Consequently, $\operatorname{ran}\mathsf J$ is invariant under $\widetilde{\mathsf S}_\mathrm{err}$, whose restriction in shell coordinates is the tridiagonal matrix
    \begin{equation}
        T_\rho \coloneqq \mathsf J^\dagger \widetilde{\mathsf S}_\mathrm{err} \mathsf J,
    \end{equation}
    with entries
    \begin{equation}\label{eq:shell-satisfaction-matrix}
        (T_\rho)_{kk}=m\rho+(1-2\rho)k, \qquad (T_\rho)_{k,k+1}=(T_\rho)_{k+1,k} = \sqrt{\rho(1-\rho)(k+1)(m-k)}.
    \end{equation}

    Next consider decoding. Let $\Pi_{\mathcal S}$ project onto the success set from \cref{eq:decoder-success-failure-partition}, and define
    \begin{equation}\label{eq:dqi-postselection-probability}
        Z_{\mathbf C}\coloneqq\sum_{\mathbf y\in\mathcal S}
            |\phi_{\mathbf C}(\mathbf y)|^2.
    \end{equation}
    For $Z_{\mathbf C}>0$, the normalized successful error and output states are
    \begin{align}
        \ket{\Gamma_{\mathcal S}}
            &\coloneqq Z_{\mathbf C}^{-1/2}\Pi_{\mathcal S}\ket\Gamma,
            \label{eq:fixed-projected-error-state}\\
        \ket{\widetilde P_{\mathcal S}}
            &\coloneqq Z_{\mathbf C}^{-1/2}\sum_{\mathbf y\in\mathcal S}
                \phi_{\mathbf C}(\mathbf y)\ket{B^T\mathbf y},
            \qquad
        \ket{P_{\mathcal S}}\coloneqq (\mathsf F_q^\dagger)^{\otimes n}
            \ket{\widetilde P_{\mathcal S}}.
            \label{eq:success-conditioned-dqi-states}
    \end{align}
    Measurement gives $\mathbb E[s(\mathbf X)]= \bra{\widetilde P_{\mathcal S}}\widetilde{\mathsf S} \ket{\widetilde P_{\mathcal S}}$.
    
    For $\mathbf y\in\mathcal S$, define
    \begin{equation}\label{eq:shift-syndrome-decode}
        F_{i,a}(\mathbf y)\coloneqq D(B^T(\mathbf y-a\mathbf e_i)),
        \qquad
        \mathcal S_{i,a}^c
        \coloneqq\{\mathbf y\in\mathcal S:\mathbf y-a\mathbf e_i \notin \mathcal S\}.
    \end{equation}
    Syndrome consistency gives $B^TF_{i,a}(\mathbf y)=B^T(\mathbf y-a\mathbf e_i)$ and $F_{i,-a}(F_{i,a}(\mathbf y))=\mathbf y$. Thus $F_{i,a}$ is a bijection of $\mathcal S$, and maps $\mathcal S_{i,a}^c$ bijectively onto $\mathcal S_{i,-a}^c$. Expanding \cref{eq:satisfaction-fourier-shifts}, the matrix element for $X_{-a\mathbf b_i}$ pairs $\mathbf z$ with $F_{i,a}(\mathbf z)$. When $\mathbf z-a\mathbf e_i\in\mathcal S$, this is precisely the error-space transition in \cref{eq:fourier-Serr-X}; the remaining pairs give
    \begin{align}
        M_{i,a}&\coloneqq \{(F_{i,a}(\mathbf z),\mathbf z): \mathbf z\in\mathcal S_{i,a}^c\}, \label{eq:pairings-set} \\
        \mathcal A_{\mathbf C} &\coloneqq \frac1{\sqrt q}\sum_{i,a} \widetilde{\mathbb1}_{C_i}(-a) \sum_{(\mathbf y,\mathbf z)\in M_{i,a}} \overline{\phi_{\mathbf C}(\mathbf y)} \phi_{\mathbf C}(\mathbf z), \label{eq:signed-collision-term} \\
        \mathbb E[s(\mathbf X)] &= \bra{\Gamma_{\mathcal S}}\widetilde{\mathsf S}_\mathrm{err} \ket{\Gamma_{\mathcal S}} + \frac{\mathcal A_{\mathbf C}}{Z_{\mathbf C}}. \label{eq:fixed-linsat-decomposition}
    \end{align}
    The first term is the expectation of the auxiliary error-space observable after restricting to successfully decoded errors. The second term records additional transitions present in syndrome space but absent from the auxiliary error-space observable. We call this contribution aliasing and control it using the one-coordinate recovery hypothesis.
    
    To bound the aliasing contribution, we express the nonzero-frequency coefficients of the target indicators in terms of the centered, normalized functions:
    \begin{equation} \label{eq:agnostic-nonlipschitz-bound}
        \frac1{\sqrt q}\widetilde{\mathbb1}_{C_i}(-a) =\sqrt{\rho(1-\rho)}\, \overline{\widetilde g_i(a)}, \qquad a\ne0.
    \end{equation}
    Define the boundary masses
    \begin{equation}\label{eq:full-shell-mass}
        \mu_{i,a}\coloneqq
            \sum_{\mathbf y\in\mathcal S_{i,a}^c}
                    |\phi_{\mathbf C}(\mathbf y)|^2.
    \end{equation}
    The zero shift contributes no aliasing. Using the coefficient identity above, Cauchy--Schwarz, and the boundary bijection gives 
    \begin{equation}\label{eq:full-shell-delta}
        |\mathcal A_{\mathbf C}|\le\Delta_{\mathbf C}
        \coloneqq\sqrt{\rho(1-\rho)}
            \sum_{i,a\ne0}|\widetilde g_i(a)|
                        \sqrt{\mu_{i,a}\mu_{i,-a}}.
    \end{equation}
    Since $\sum_{a\neq 0} \abs{\widetilde{g}_i(a)}^2 = 1$, Cauchy--Schwarz also gives
    \begin{equation} \label{eq:centered-fourier-l1-bound}
        \sum_{a\ne0}|\widetilde g_i(a)|\le\sqrt{q-1}.
    \end{equation}

    Let $\epsilon_k$ be the decoder's ordinary failure probability under $P_k$. The hypothesis gives $\epsilon_k\le\varepsilon$ for $k\le\ell$. For every $k\le\ell$, the one-coordinate recovery hypothesis implies $P_k(\mathcal S_{i,a}^c)\le\varepsilon$. The prepared error has weight $k$ with probability $|w_k|^2$ and, conditional on that weight, has distribution $P_k$. Therefore, for every normalized amplitude vector supported on weights $0,\ldots,\ell$,
    \begin{equation}
        \mu_{i,a}=\sum_{k=0}^{\ell}|w_k|^2P_k(\mathcal S_{i,a}^c) \le \varepsilon\sum_{k=0}^{\ell}|w_k|^2=\varepsilon.
    \end{equation}
    Substituting this estimate into the boundary-mass bound and applying Eq.~\eqref{eq:centered-fourier-l1-bound} gives
    \begin{equation}\label{eq:cutoff-alias-bound}
        |\mathcal A_{\mathbf C}| \leq \varepsilon \sqrt{\rho(1-\rho)} \sum_{i,a\ne0}|\widetilde{g}_i(a)| \leq m\varepsilon \sqrt{(q-1)\rho(1-\rho)}.
    \end{equation}
    
    We next bound the effect of postselection on the auxiliary expectation. Let $T_{\rho,\ell}$ be the principal block of $T_\rho$ on weights $0,\ldots,\ell$, and let $\lambda_\ell$ be its largest eigenvalue. In the shell basis, the compression of the success projector is diagonal, with entries equal to the decoding success probabilities on the individual shells.
    \begin{equation}\label{eq:cutoff-success-compression}
        Q_{\mathcal S}\coloneqq\mathsf J^\dagger\Pi_{\mathcal S}\mathsf J = \diag(1-\epsilon_0,\ldots,1-\epsilon_m).
    \end{equation}
    For normalized $\mathbf w$ supported through $\ell$, decompose
    \begin{equation}\label{eq:cutoff-success-decomposition}
        \Pi_{\mathcal S}\mathsf J\mathbf w = \mathsf J Q_{\mathcal S}\mathbf w+\ket\eta, \qquad
        \ket\eta\perp\operatorname{ran}\mathsf J, \qquad \|\eta\|^2=Z_{\mathbf C}-\|Q_{\mathcal S}\mathbf w\|^2.
    \end{equation}
    Here $Z_{\mathbf C}=\mathbf w^\dagger Q_{\mathcal S}\mathbf w \ge1-\varepsilon$. Invariance makes the cross terms under $\widetilde{\mathsf S}_\mathrm{err}$ vanish. Since $\|Q_{\mathcal S}\mathbf w\|^2\ge(1-\varepsilon)Z_{\mathbf C}$, every allowed $\mathbf w$ satisfies
    \begin{equation}\label{eq:cutoff-projection-upper}
        \left\langle\widetilde{\mathsf S}_\mathrm{err}\right\rangle_{\Gamma_{\mathcal S}}
        \le\lambda_\ell\frac{\|Q_{\mathcal S}\mathbf w\|^2}{Z_{\mathbf C}}
            +m\left(1-\frac{\|Q_{\mathcal S}\mathbf w\|^2}{Z_{\mathbf C}}\right)
        \le(1-\varepsilon)\lambda_\ell+m\varepsilon.
    \end{equation}
    For the lower bound, choose $\mathbf w$ to be a top eigenvector of $T_{\rho,\ell}$. The component of $Q_{\mathcal S}\mathbf w$ along $\mathbf w$ is $Z_{\mathbf C}\mathbf w$. Since $T_{\rho,\ell}\succeq0$, orthogonality gives $\mathbf w^\dagger Q_{\mathcal S}T_{\rho,\ell}Q_{\mathcal S}\mathbf w \ge\lambda_\ell Z_{\mathbf C}^2$, with the matrices restricted to the occupied shells. The contribution from $\ket{\eta}$ is also nonnegative, since $\widetilde{\mathsf S}_\mathrm{err} \succeq 0$, so
    \begin{equation}\label{eq:cutoff-projection-lower}
        \left\langle\widetilde{\mathsf S}_\mathrm{err}\right\rangle_{\Gamma_{\mathcal S}}
        \ge\lambda_\ell Z_{\mathbf C}\ge(1-\varepsilon)\lambda_\ell.
    \end{equation}
    Combining \cref{eq:fixed-linsat-decomposition,eq:cutoff-alias-bound,eq:cutoff-projection-upper,eq:cutoff-projection-lower} proves
    \begin{equation}\label{eq:cutoff-spectral-bounds}
    \begin{aligned}
        f_\ell^* &\ge (1-\varepsilon)\frac{\lambda_\ell}{m} - \frac{\sqrt{(q-1)\rho(1-\rho)}\,\varepsilon}{1-\varepsilon},\\
        f_\ell^* &\le (1-\varepsilon)\frac{\lambda_\ell}{m}
            +\varepsilon+\frac{\sqrt{(q-1)\rho(1-\rho)}\,\varepsilon}{1-\varepsilon}.
    \end{aligned}
    \end{equation}
    
    Jordan et al.~\cite[\S 9]{JSW25} determined the asymptotic behavior of the largest eigenvalue $\lambda_\ell$ of $T_{\rho,\ell}$. We give an elementary proof of the following uniform non-asymptotic refinement:
    \begin{equation}\label{eq:cutoff-semicircle-edge}
        F_\rho(\ell/m)-m^{-1/2} \le \frac{\lambda_\ell}{m} \le F_\rho(\ell/m).
    \end{equation}
    For the upper bound, let $\mathbf w$ be normalized and supported on weights $0,\ldots,\ell$, and put $x=m^{-1}\sum_k k|w_k|^2\le\ell/m$. Positivity bounds the magnitude of the off-diagonal entry of the one-qubit reduced density matrix of $\sum_kw_k\ket{D_k^m}$ by $\sqrt{x(1-x)}$. \Cref{eq:Pi-tilde-basis} then gives $\mathbf w^\dagger T_\rho\mathbf w/m\le F_\rho(x)\le F_\rho(\ell/m)$.
    
    For the lower bound, it suffices to consider $\ell\geq 1$, since $\lambda_0/m = \rho = F_\rho(0)$. Set $p_\ell \coloneqq \min \{\ell/m, 1-\rho\}$, so that $F_\rho(p_\ell) = F_\rho(\ell/m)$. Consider the binomial amplitudes
    \begin{equation}\label{eq:binomial-shell-vector}
        w_k^{(p)}\coloneqq\sqrt{\binom mk p^k(1-p)^{m-k}},
        \qquad 0\le k\le m,\quad p\in[0,1].
    \end{equation}
    The corresponding qubit state is $(\sqrt{1-p}\ket0+\sqrt p\ket1)^{\otimes m}$. For $p\leq 1-\rho$, \cref{eq:Pi-tilde-basis} gives
    \begin{equation}\label{eq:binomial-jacobi-benchmark}
        \frac{(\mathbf w^{(p)})^\dagger T_\rho\mathbf w^{(p)}}{m}=F_\rho(p).
    \end{equation}
    Truncate $\mathbf w^{(p_\ell)}$ at weight $\ell$ and normalize. The variational principle gives
    \begin{equation}\label{eq:cutoff-binomial-rayleigh}
        \frac{\lambda_\ell}{m}\ge
            \rho+(F_\rho(p_\ell)-\rho)R_\ell,
        \qquad
        R_\ell\coloneqq
            \frac{\Pr[\operatorname{Bin}(m-1,p_\ell)\le\ell-1]}
                 {\Pr[\operatorname{Bin}(m,p_\ell)\le\ell]}.
    \end{equation}
    Indeed, by \cref{eq:shell-satisfaction-matrix}, truncation multiplies both the normalized mean weight and the adjacent-shell contribution by $R_\ell$. In particular, for $K\sim\operatorname{Bin}(m,p_\ell)$, we have $\mathbb E[K\mid K\le\ell]=mp_\ell R_\ell$. Since $\ell\ge mp_\ell$, the binomial median bound of~\cite[Theorem~3]{KB78} gives $\Pr[K\le\ell]\ge1/2$. Thus
    \begin{equation}\label{eq:cutoff-binomial-mean-loss}
        mp_\ell(1-R_\ell)
        \le\frac{\mathbb E[(mp_\ell-K)_+]}{\Pr[K\le\ell]}
        \le\sqrt{mp_\ell(1-p_\ell)},
    \end{equation}
    where the last step uses $\mathbb E[(mp_\ell-K)_+]=\tfrac12\mathbb E|K-mp_\ell|$ and Cauchy--Schwarz. Consequently,
    \begin{equation}\label{eq:cutoff-finite-size-loss}
    \begin{aligned}
        (F_\rho(p_\ell)-\rho)(1-R_\ell)
        &\le\frac{(1-2\rho)\sqrt{p_\ell(1-p_\ell)}
            +2\sqrt{\rho(1-\rho)}(1-p_\ell)}{\sqrt m}\le m^{-1/2},
    \end{aligned}
    \end{equation}
    again by Cauchy--Schwarz and $(1-2\rho)^2+4\rho(1-\rho)=1$. Combining \cref{eq:cutoff-binomial-rayleigh,eq:cutoff-finite-size-loss} proves \cref{eq:cutoff-semicircle-edge}. Substituting this bound into \cref{eq:cutoff-spectral-bounds} completes the proof.
\end{proof}

For independently random targets and a decoder fixed independently of them, aliasing vanishes after averaging. With binomial amplitudes, this yields a guarantee involving only the decoder’s ordinary word-error probability, without the one-coordinate recovery requirement.
\begin{corollary}[Independently random target sets]
\label{cor:DQI-error-random-binomial-weights}
    Fix a constraint matrix with $m$ rows and a deterministic, syndrome-consistent decoder that returns zero on the zero syndrome. Choose the target sets independently and uniformly among the $r$-element subsets of $\mathbb F_q$, where $0<r<q$, independently of the decoder. Let $\varepsilon_p$ be the decoder's word-error probability when each error coordinate is independently zero with probability $1-p$ and otherwise uniformly distributed over the nonzero field elements. For each target tuple, condition on successful postselection, and let $f$ be the resulting expected satisfaction fraction averaged over the target tuples. Then, for DQI with binomial amplitudes \cref{eq:binomial-shell-vector} with $p \in [0,1-\rho]$,
    \begin{equation}\label{eq:random-linsat-binomial-bound}
    \begin{gathered}
        (1-\varepsilon_p)F_\rho(p)-v \le f
        \le (1-\varepsilon_p)F_\rho(p)+\varepsilon_p+v,\qquad v \coloneqq 2\sqrt{\frac{F_\rho(p)(1-F_\rho(p))\varepsilon_p(1-\varepsilon_p)}m}.
    \end{gathered}
    \end{equation}
    In particular, $v\le 0.5 \,m^{-1/2}$.
\end{corollary}
\begin{proof}
The binomial amplitudes give the product state
\begin{equation}\label{eq:binomial-product-state}
 \ket{\Gamma^{(p)}}=
 \bigotimes_{i=1}^m
 \left(\sqrt{1-p}\ket0+
       \sqrt p\sum_{a\ne0}\widetilde g_i(a)\ket a\right).
\end{equation}
Its expectation and variance under $\widetilde{\mathsf S}_\mathrm{err}$ are
\begin{equation}\label{eq:binomial-satisfaction-variance}
 \mathbb E_{\Gamma^{(p)}}\widetilde{\mathsf S}_\mathrm{err}
   =mF_\rho(p), \qquad \Var_{\Gamma^{(p)}}
      (\widetilde{\mathsf S}_\mathrm{err})
   =mF_\rho(p)(1-F_\rho(p)).
\end{equation}
Indeed, each local term is a projector with expectation $F_\rho(p)$ by \cref{eq:Pi-tilde-basis}, and the covariances vanish.

For a fixed target tuple, abbreviate $Z=Z_{\mathbf C}$. The assumption on the zero syndrome and the nonzero amplitude at the zero error ensure $Z>0$. Write
\begin{equation}
    \ket{\Gamma_{\mathcal S}^{(p)}} = \sqrt Z\ket{\Gamma^{(p)}}+\sqrt{1-Z}\ket\eta, \qquad \braket{\eta\mid\Gamma^{(p)}}=0,\quad\|\eta\|=1.
\end{equation}
Any expectation value of $\widetilde{\mathsf S}_\mathrm{err}$ necessarily lies in $[0,m]$. Subtracting $m F_\rho(p)I$ in the cross term and applying Cauchy--Schwarz therefore gives
\begin{equation}\label{eq:binomial-projected-satisfaction}
\begin{aligned}
    ZF_\rho(p)-2\sqrt{\frac{Z(1-Z)F_\rho(p)(1-F_\rho(p))}m}
    &\le\frac1m \left\langle\widetilde{\mathsf S}_\mathrm{err} \right\rangle_{\Gamma_{\mathcal S}^{(p)}}\\
    &\le ZF_\rho(p)+1-Z + 2\sqrt{\frac{Z(1-Z)F_\rho(p)(1-F_\rho(p))}m}.
\end{aligned}
\end{equation}

It remains to average the alias term in \cref{eq:fixed-linsat-decomposition}. Translating each target independently by $v_i$ preserves the target distribution and changes the error amplitudes only by phases, leaving $Z_\mathbf C$ unchanged. In \cref{eq:signed-collision-term}, each summand is multiplied by $\chi_q(\mathbf v\cdot(\mathbf z-\mathbf y-a\mathbf e_i))$. The vector in this character is nonzero by \cref{eq:pairings-set}, so averaging the translations makes $\mathbb E_{\mathbf C}[\mathcal A_{\mathbf C}/Z_{\mathbf C}]=0$.

Finally, uniform target sets satisfy $\mathbb E_{C_i}|\widetilde g_i(a)|^2=1/(q-1)$ for $a\ne0$. Their independence implies $\mathbb E_{\mathbf C}(1-Z)=\varepsilon_p$. Averaging \cref{eq:binomial-projected-satisfaction} and using $\mathbb E_{\mathbf C}\sqrt{Z(1-Z)} \le\sqrt{(1-\varepsilon_p)\varepsilon_p}$ proves the result.
\end{proof}

These results translate decoding guarantees under the DQI-induced error distributions into satisfaction guarantees. We next establish the required decoding properties for cycle codes.

\section{Decoding cycle codes over general fields} \label{sec:decoding-cycle-codes}

Let $G$ be an undirected, simple graph with vertex set $V$ and edge set $E$. To define a homological cycle code from $G$, fix an orientation of each edge $e$, and a nonzero weight $a_e\in\mathbb{F}_q^\times$. Let $B^T\in\mathbb{F}_q^{V\times E}$ be the weighted incidence matrix in which column $e$, oriented from $u$ to $v$, has two nonzero entries: $a_e$ in row $v$ and $-a_e$ in row $u$. The cycle code $C^\perp$ is the code whose parity-check matrix is $B^T$. 
\begin{equation}    \label{eq:cycle-code-definition}
    C^\perp\coloneqq \ker(B^T)=\{\mathbf{d}\in\mathbb{F}_q^E:B^T\mathbf{d}=\mathbf{0}\}
\end{equation}

One can visualize homological cycle codes as incompressible flows on graph $G$. Observe that a vector $\mathbf{d}\in\mathbb{F}_q^E$ satisfies $B^T\mathbf{d}=\mathbf{0}$ exactly when the weighted sum of the edge values at every vertex vanishes in $\mathbb F_q$. That is, the net flux into each vertex is zero. Linearity of the code can be understood by noting that sums of incompressible flows are also incompressible. The name ``cycle code'' reflects the fact that this space is spanned by the weighted circulation vectors of the simple cycles of $G$.

For homological cycle codes, the minimum distance $d_\mathrm{min}(C^\perp)$ is fixed by the geometry of the underlying graph. Since $C^\perp$ is linear, $d_\mathrm{min}$ is exactly the smallest weight of a nonzero codeword. Because every nonzero codeword $\mathbf d$ is a valid flow, its associated subgraph $\supp \mathbf d$ has no vertex of degree one, since zero-flux would be violated at that vertex. Hence each nonzero codeword contains a cycle of $G$, implying $d_\mathrm{min}\geq g_G$. Conversely, any single cycle carries a valid flow supported on exactly its edges. By considering the flow along the shortest cycle, we have the complementary bound $d_\mathrm{min} \leq g_G$. Thus, $d_\mathrm{min} = g_G$ for cycle codes, independent of $q$. Because of this, cycle codes on the bounded-degree graph families most useful for decoding cannot have large distance. The Moore bound~\cite{Biggs1993AlgebraicGraphTheory} forces the girth of a $D$-regular graph, for any fixed $D$, to grow at most logarithmically in the number of vertices, and hence so does the distance. Remarkably, however, families of cycle codes often are still able to correct, with high probability, a number of errors that grows linearly with block size, thus demonstrating a significant gap between worst-case and probabilistic error correction capabilities.

\subsection{NP-hardness of minimum-weight decoding}
\label{sec:ml-hardness}

An important task in coding theory (and practice) is determining a nearest codeword $\mathbf d$ to some given string $\mathbf y$, which may be thought of as the output of some noisy communication channel.
\begin{definition}[Nearest Codeword Problem] \label{defn:nearest-codeword}
    Given a block code $C^\perp\subseteq\mathbb{F}_q^m$ and a string $\mathbf{y}\in\mathbb{F}_q^m$, the nearest codeword problem is to find the codeword $\mathbf{d}\in C^\perp$ nearest to $\mathbf{y}$ in Hamming distance.
\end{definition}
\noindent When $C^\perp$ is linear, this problem is equivalent to minimum-weight syndrome decoding with respect to its parity-check matrix $B^T$: writing $\mathbf{s}=B^T\mathbf{y}$ for the syndrome, finding the nearest codeword $\mathbf d$ is equivalent to finding a minimum-weight $\mathbf{e}\in\mathbb{F}_q^m$ satisfying $B^T\mathbf{e}=\mathbf{s}$, via the relation $\mathbf e = \mathbf y - \mathbf d$. We now prove the NP-hardness of nearest codeword in this syndrome-decoding formulation for cycle codes on alphabets of size three or more. This stands in contrast to cycle codes on the binary alphabet, for which nearest-codeword decoding is known to be solvable in polynomial time by a $T$-join decoding algorithm \cite{EJ73}.

\begin{theorem}\label{thm:hardness}
    For every fixed finite field $\mathbb F_q$ with $q>2$, the following decision problem is NP-complete. Given a homological cycle code with parity-check matrix $B^T$, a syndrome $\mathbf s$, and an integer $w\ge0$, decide whether there exists $\mathbf e\in\mathbb F_q^m$ with $B^T\mathbf e=\mathbf s$ and $|\mathbf e|\le w$. Consequently, nearest-codeword decoding is NP-hard for these codes.
\end{theorem}
\begin{proof}
    Membership in NP follows by checking the syndrome and weight. For hardness, we reduce from perfect three-dimensional matching~\cite{GJ79}: given disjoint sets $U_1,U_2,U_3$, each of size $k$, and a set of triples $\mathcal T\subseteq U_1\times U_2\times U_3$, decide whether there exists a subset of $k$ triples from $\mathcal{T}$ that covers every element exactly once.
    
    Choose $\alpha,\beta\in\mathbb F_q^\times$ with $\alpha+\beta\ne0$ (possible because $q>2$) and put $\gamma=-\alpha-\beta$. All three are nonzero, their sum is zero, and no nonempty proper subcollection has sum zero. This holds also in characteristic two.
    
    Construct a bipartite incidence graph with vertex set $V = P \sqcup L$, with ``element vertices'' $ P = U_1\sqcup U_2\sqcup U_3$, and ``triplet vertices'' $L = \mathcal T$. Join each triplet vertex $T\in \mathcal T$ to its three elements, orienting the edges from $T$ to those elements, and take $B^T$ to be its oriented incidence matrix. Define a syndrome $\mathbf s \in \mathbb{F}_q^V$ by assigning entries in $U_1,U_2,U_3$ values $\alpha,\beta,\gamma$, respectively, and entries in $\mathcal{T}$ value $0$. Set $w=3k$.
    
    An exact cover yields an error $\mathbf e\in\mathbb{F}_q^E$ of weight $3k$ by assigning the three edges of each chosen triple the respective values $\alpha,\beta,\gamma$. The triplet-vertex equations hold because their sum is zero. Conversely, all $3k$ element syndromes are nonzero, so an error of weight at most $3k$ has exactly one nonzero edge at each element. That edge's value is forced to equal the element syndrome. At a triplet vertex, the selected incident values are a subcollection of $\{\alpha,\beta,\gamma\}$, which sums to zero only if none or all three are selected. The selected triples therefore form an exact cover. This reduction gives a simple graph.
\end{proof}

\subsection{Decoding thresholds for even symbol-wise additive channels}
\label{sec:ml-threshold}

We now study ML whole-word decoding for the homological cycle codes on $D$-regular graphs. Throughout this subsection, $B^T$ is an oriented incidence matrix, $n=|V(G)|$, and $m=|E(G)|=Dn/2$. For a weighted homological incidence matrix, these results apply after multiplying each error coordinate by its incidence weight, provided the transformed errors have the common distribution specified below. Uniform nonzero replacement is preserved by this normalization.

We consider additive channels in which each symbol is corrupted by adding an error $a \in \mathbb{F}_q$, whose distribution is as follows. Fix a probability distribution $\theta$ on $\mathbb F_q^\times$ which is \emph{even}, i.e., $\theta(a)=\theta(-a)$, and write
\begin{equation}\label{eq:ml-general-noise}
    \nu_p(0) = 1-p, \qquad \nu_p(a)=p\theta(a)\quad(a\ne0)
\end{equation}
We refer to $\theta$ as the \emph{replacement distribution}. An important special case, when all symbols are equally likely to occur given a flip, is the $q$-ary symmetric channel.
\begin{definition}[$q$-ary symmetric channel]\label{def:qsym}
The $q$-ary symmetric channel ($q$SC) has uniform replacement $\theta(a)=1/(q-1)$ for all $a\ne0$ and $0\le p\le 1-q^{-1}$.
\end{definition}

Let $\mathbf E \in \mathbb{F}_q^m$ denote the random variable representing the full error vector obtained by drawing $m$ independent errors from $\nu_p$. Let $P_{\mathrm W}^{\mathrm{ML}}(G,p,\theta)$ be the minimum whole-word error probability in recovering $\mathbf E$ from $B^T\mathbf E$, and put $P_{\mathrm S}^{\mathrm{ML}}=1-P_{\mathrm W}^{\mathrm{ML}}$. Thus,
\begin{equation}\label{eq:ml-success}
 P_{\mathrm S}^{\mathrm{ML}}(G,p,\theta)
 =\sum_{\mathbf s\in \im B^T}
   \max_{\substack{\mathbf e\in\mathbb F_q^E\\B^T\mathbf e=\mathbf s}}
   \prod_{j=1}^{m}\nu_p(e_j).
\end{equation}
This is also ML codeword success for a uniformly distributed transmitted codeword, and it is independent of the choice among equal-likelihood maximizers. For uniform $\theta$, we retain the notation $P_{\mathrm W}^{\mathrm{ML}}(G,p)$. Additionally, for uniform $\theta$ and $p<1-q^{-1}$, the likelihood of an error decreases with Hamming weight, so ML syndrome decoding is minimum-weight decoding. Its worst-case hardness for $q>2$ does not preclude efficient recovery at lower noise levels, as the LP guarantees below demonstrate.

From the Moore bound \cite{Biggs1993AlgebraicGraphTheory}, the girth for $D$-regular graphs can be at most logarithmic in terms of the number of vertices, as given by
\begin{equation}
    g_G \leq 2 \log_{D-1}(n) + O(1)~.
\end{equation}
For random regular graphs, the counts of cycles of each fixed length converge to Poisson distributions \cite{mckay2004}. The presence of these short cycles hurts the decoding performance, leading to a word error rate bounded away from zero. 
\begin{theorem}[Error floor from short cycles]
\label{thm:unexpurgated-error-floor}
    Fix $D \ge 3$, a finite field $\mathbb F_q$, and an even replacement distribution $\theta$. Let $\theta_\mathrm{max} \coloneqq \max_{a} \theta(a)$, and suppose $0<p<(1+\theta_\mathrm{max})^{-1}$. If $G_n$ is uniform among simple $D$-regular graphs on $n$ labeled vertices, then
    \begin{equation}\label{eq:ml-unexpurgated-floor}
     \liminf_{n\to\infty}
     \mathbb E_{G_n}P_{\mathrm W}^{\mathrm{ML}}(G_n,p,\theta)
     \ge (p \theta_\mathrm{max})^3
           \left(1-\exp\left[-\frac{(D-1)^3}{6}\right]\right)>0.
    \end{equation}
    The limit is through $n>D$ with $Dn$ even. The same statement holds after conditioning on connectedness.
\end{theorem}
\begin{proof}
    Whenever $G_n$ contains a triangle, select one and choose an orientation. Let $\mathbf f \in \mathbb F_q^E$ be the associated signed unit circulation. Choose $a\ne0$ with $\theta(a)=\theta_\mathrm{max}$. With probability $(p\theta_\mathrm{max})^3$, the error equals $a\mathbf f$ on the triangle, because $\theta$ is even. The assumed range of $p$ gives $1-p>p\theta_\mathrm{max}$. Subtracting this circulation leaves the syndrome and all other coordinates unchanged, while replacing three nonzero entries, each of probability $p\theta_\mathrm{max}$, by zero entries, each of probability $1-p$. Thus the resulting error is strictly more likely, and every ML decoder fails on the original error. Therefore,
    \begin{equation}
     P_{\mathrm W}^{\mathrm{ML}}(G_n,p,\theta)
     \ge(p\theta_\mathrm{max})^3
           \mathbb1_{\{G_n\text{ contains a triangle}\}}.
    \end{equation}
    The triangle count converges to a Poisson random variable with mean $(D-1)^3/6$~\cite{mckay2004,Johnson15}. Averaging the preceding inequality and taking the lower limit proves \cref{eq:ml-unexpurgated-floor}. The probability of connectedness tends to one~\cite{Bollobas2001RandomGraphs}, so conditioning on connectedness leaves the same limit.
\end{proof} 
\noindent For uniform replacement, $\theta_\mathrm{max} = (q-1)^{-1}$, so the theorem applies throughout $0<p<1-q^{-1}$. Excluding cycles only up to any fixed length still leaves a positive error floor, by the joint Poisson limit for the remaining fixed cycle lengths~\cite{mckay2004,Johnson15}. We therefore employ an ensemble with girth growing logarithmically.

\begin{definition}[Linial--Simkin ensemble]
\label{def:linial-simkin-ensemble}
Fix $D\ge3$, $c\in(0,1)$, and even $n$. Set
\begin{equation}\label{eq:ml-girth-schedule}
 g(n)\coloneqq \max\{3,\lfloor c\log_{D-1}n\rfloor\}.
\end{equation}
Begin with a fixed Hamilton cycle $H$ on $[n]$. When the current minimum degree is $k-1<D$, choose an unordered pair of degree-$(k-1)$ vertices uniformly at random from those at current graph distance at least $g(n)-1$, and add the edge between them. Abort if no such pair exists; otherwise continue until the graph is $D$-regular. A completed degree-increment round adds a perfect matching, so a successful run has the form
\begin{equation}\label{eq:ls-decomposition}
    G = H \cup M_1 \cup \cdots \cup M_{D-2}.
\end{equation}
Let $\mathcal L_{n,D,c}$ be the output probability distribution conditioned on successful completion. 
\end{definition}
\noindent Every successful output $G$ of the Linial--Simkin process is simple, connected, and satisfies $g_G \geq g(n)$. The completion probability tends to one by~\cite[Theorem~1.1 and Proposition~2.1]{LS}, while the uniform choice of admissible pairs is part of the definition.

We now turn to the ML decoding properties of the Linial--Simkin graphs. The key point is that the syndrome determines the error only modulo the
cycle space:
\begin{equation}
       B^T\mathbf e=B^T\mathbf e'
   \qquad\Longleftrightarrow\qquad
   \mathbf e-\mathbf e'\in\ker B^T.
\end{equation}
Thus every nonzero circulation $\mathbf f\in\ker B^T$ gives a possible competitor $\mathbf E-\mathbf f$ to the true error $\mathbf E$. ML decoding succeeds precisely when the true error is the most likely element of its syndrome coset (ignoring ties). Consequently, the basic obstruction to ML decoding is a circulation $\mathbf f\ne0$ for which $\mathbf E-\mathbf f$ is at least as likely as $\mathbf E$. For a fixed circulation $\mathbf f$, this pairwise confusion probability
is controlled by the Bhattacharyya overlap between the probability distribution over noise vectors and its coordinatewise translate. Indeed, writing
\begin{equation}
    L(\mathbf e)\coloneqq \prod_{j=1}^m \nu_p(e_j),
\end{equation}
the elementary inequality
\begin{equation}
   \mathbb1_{\{L(\mathbf e-\mathbf f)\ge L(\mathbf e)\}}L(\mathbf e)
   \le \sqrt{L(\mathbf e)L(\mathbf e-\mathbf f)}
\end{equation}
gives
\begin{equation}
 \Pr\{L(\mathbf E-\mathbf f)\ge L(\mathbf E)\}
 \le
 \sum_{\mathbf e\in\mathbb F_q^E}
   \sqrt{L(\mathbf e)L(\mathbf e-\mathbf f)}
 =
 \prod_{j:f_j\ne0}
   \sum_{z\in\mathbb F_q}
      \sqrt{\nu_p(z)\nu_p(z-f_j)} .
\end{equation}
Thus the worst single-coordinate overlap controls the pairwise ML error probability for every competing circulation, with a cost exponential in the support size of that circulation. This is the same channel parameter that appears in standard pairwise ML bounds for group codes on symmetric channels~\cite[\S 2.5]{Como08}.

The competing objects here are cycles and, more generally, circulations. A $D$-regular graph has approximately $(D-1)^\ell$ possible non-backtracking closed walks of length $\ell$, while each nonzero
coordinate of a competing circulation contributes a factor at most $Z_\theta(p)$ to the pairwise overlap, where
\begin{equation}\label{eq:ml-general-Z}
    Z_\theta(p)\coloneqq \max_{a\ne0}\sum_{z\in\mathbb F_q} \sqrt{\nu_p(z)\nu_p(z-a)},
\end{equation}
or equivalently 
\begin{equation}\label{eq:ml-general-Z-expanded}
    Z_\theta(p)=\max_{a\ne0}
    \left\{2\sqrt{p(1-p)\theta(a)}
    +p\sum_{z\notin\{0,a\}}\sqrt{\theta(z)\theta(z-a)}\right\}, 
\end{equation}
which is non-decreasing for $0\le p\le1/2$. Hence, the natural branching parameter is
\begin{equation}
   (D-1)Z_\theta(p).
\end{equation}
When this quantity is below one, the total contribution of long cycle-like competitors is summable, and reliability can hold. When it is above one, the graph contains exponentially many sufficiently independent cycle competitors, and their aggregate likelihood overwhelms the true error, giving a strong converse. This motivates the following definition.

\begin{equation}\label{eq:ml-general-cycle-value}
    \delta_\mathrm{cyc}(D,\theta) \coloneqq \inf\{p\in[0,1/2]:(D-1)Z_\theta(p)\ge1\},
\end{equation}
with value $1/2$ if this set is empty.

We now prove the condition under which the whole-word decoding error rate goes to unity or vanishes. Define the \emph{reliability threshold}
\begin{equation}\label{eq:ml-reliability-threshold}
    \delta_\mathrm{rel}(D,\theta) \coloneqq \sup\left\{p\in[0,1/2]: \lim_{n\rightarrow \infty} \mathbb E_{G\sim\mathcal L_{n,D,c}}
    P_{\mathrm W}^{\mathrm{ML}}(G,p,\theta) = 0\right\}.
\end{equation}
The limit is taken with $n$ even and other parameters fixed. For uniform $\theta$, write $\delta_\mathrm{rel}(D,q)$.
\begin{theorem}[ML strong converse]
\label{thm:ml-general-converse}
Fix $D\ge3$, a finite field $\mathbb F_q$, $c\in(0,1)$, and an even replacement distribution $\theta$. For the Linial--Simkin ensemble and independent errors with distribution \cref{eq:ml-general-noise},
\begin{equation}\label{eq:ml-general-strong}
    (D-1)Z_\theta(p)>1 \quad\implies\quad \mathbb E_G P_{\mathrm S}^{\mathrm{ML}}(G,p,\theta)\rightarrow0, \qquad 0<p\le1/2.
\end{equation}
In particular, $\delta_\mathrm{rel}(D,\theta)\le\delta_\mathrm{cyc}(D,\theta)$. For uniform replacement, the same strong converse holds throughout $\delta_\mathrm{cyc}(D,q)<p\le 1-q^{-1}$.
\end{theorem}

\begin{theorem}[Small-alphabet ML reliability]
\label{thm:ml-small-alphabet}
Fix $c\in(0,1)$ and an even replacement distribution $\theta$ on $\mathbb F_q^\times$. For the Linial--Simkin ensemble, if $\gamma=q/D^{2/3}\le1/32$, then
\begin{equation}\label{eq:ml-small-bracket}
    (1-16\gamma)\delta_\mathrm{cyc}(D,\theta) \le \delta_\mathrm{rel}(D,\theta) \le \delta_\mathrm{cyc}(D,\theta).
\end{equation}
Throughout the lower interval, a minimum-Hamming-weight syndrome decoder recovers the errors with probability tending to one, as does the optimal ML decoder. As $D\to\infty$, uniformly over even $\theta$ and finite fields with $q=o(D^{2/3})$,
\begin{equation}\label{eq:ml-small-asymptotic}
 \frac{\delta_\mathrm{rel}(D,\theta)}
      {\delta_\mathrm{cyc}(D,\theta)}\rightarrow1,
 \qquad
 \delta_\mathrm{rel}(D,\theta)
 =\frac{1+o(1)}{4\theta_\mathrm{max}D^2}.
\end{equation}
\end{theorem}
Both proofs are in \cref{app:ml-small-alphabet}. The recovery proof also shows that, with probability tending to one, every word differing from the error in at most one coordinate is uniquely recovered by the minimum-weight decoder; see \cref{lem:robust-minweight-recovery}.

For uniform $\theta$, \cref{eq:ml-general-Z,eq:ml-general-cycle-value} reduce to
\begin{align}
    Z_q(p) &= 2\sqrt{\frac{p(1-p)}{q-1}}+\frac{q-2}{q-1}p,
    \label{eq:ml-Z}\\
    \delta_\mathrm{cyc}(D,q) &= \frac{q-1}{(D-1)(\sqrt{D+q-2}+\sqrt{D-2})^2}. \label{eq:ml-cycle-value}
\end{align}
Since $Z_q(1/2)>1/2\ge1/(D-1)$, this crossing lies below $1/2$; the full-range strong converse therefore justifies the notation $\delta_\mathrm{rel}(D,q)$ used above. Thus \cref{thm:ml-small-alphabet} gives $\delta_\mathrm{rel}(D,q)\sim(q-1)/(4D^2)$ when $q=o(D^{2/3})$. For $q=2$, the reliability endpoint is exact at every $D\ge3$:
\begin{equation}\label{eq:ml-binary-exact}
 \delta_\mathrm{rel}(D,2)=\delta_\mathrm{cyc}(D,2)
 =\frac12\left(1-\sqrt{1-\frac1{(D-1)^2}}\right).
\end{equation}

\subsection{Decoding with Linear Programming}
\label{sec:lp-decoder}

Although minimum-weight decoding is NP-hard for $q>2$, a polynomial-time decoder can still recover from a positive fraction of randomly located symbol errors. We establish this for the standard nonbinary local-polytope LP decoder~\cite{FSBG09} on the Linial--Simkin ensemble. The recovery bound uses only the graph's degree and girth but is independent of the alphabet size.

\subsubsection{The decoder}
Let $H=B^T$, write $h_{ve}=H_{ve}$, and let $\partial v$ be the edges incident to $v$. Given a syndrome $\mathbf s\in\operatorname{im}H$, define the local constraint set
\begin{equation}
 \mathcal C_v(s_v)\coloneqq\left\{(b_e)_{e\in\partial v}\in\mathbb F_q^{\partial v}:
       \sum_{e\in\partial v}h_{ve}b_e=s_v\right\}.
\end{equation}
Instead of choosing a single error value on each edge $e$, the LP assigns a probability $x_{e,a}$ to each possible value $a\in\mathbb F_q$. These real-valued variables are called \emph{symbol marginals}; write $\mathbf x$ for the vector of all such variables. At each vertex $v$, the LP also chooses a probability distribution $\lambda_{v,\mathbf b}$ over the assignments in $\mathcal C_v(s_v)$. The two local distributions involving an edge must give the same marginal probabilities for its value. The constraints are
\begin{equation}\label{eq:lp-local-polytope}
\begin{aligned}
 &x_{e,a}\ge0,\qquad \sum_{a\in\mathbb F_q}x_{e,a}=1,\\
 &\lambda_{v,\mathbf b}\ge0,\qquad
       \sum_{\mathbf b\in\mathcal C_v(s_v)}\lambda_{v,\mathbf b}=1,\\
 &x_{e,a}=\sum_{\substack{\mathbf b\in\mathcal C_v(s_v)\\b_e=a}}
                 \lambda_{v,\mathbf b}
       \qquad(v\in V,\ e\in\partial v,\ a\in\mathbb F_q).
\end{aligned}
\end{equation}
The local distributions need not come from a joint distribution on complete words satisfying all checks. The LP minimizes the relaxed Hamming weight
\begin{equation}\label{eq:lp-objective}
 \mathcal L(\mathbf x)=\sum_{e\in E}(1-x_{e,0}).
\end{equation}
We call the marginal vector $\mathbf x$ \emph{integral} when every $x_{e,a}$ is either zero or one, and \emph{fractional} otherwise. Every syndrome-consistent word $\mathbf y$ gives an integral feasible point $x_{e,a}=\mathbb1\{a=y_e\}$, with objective value $|\mathbf y|$. Conversely, normalization and local consistency ensure that every integral feasible marginal vector represents a syndrome-consistent word. An integral optimum therefore gives a minimum-weight word with the specified syndrome, and hence an ML correction on the $q$SC~\cite{FSBG09}.

If the LP has several optimal solutions, select one using a fixed deterministic tie-breaking rule. If the selected symbol marginals are integral, output the corresponding word; otherwise declare failure. Let $P_{\mathrm W}^{\mathrm{LP}}(G,p)$ be the probability of either a declared failure or an incorrect output when each position is corrupted independently with probability $p$; the nonzero values may have any distribution. The recovery condition below guarantees that the true error has the unique optimal marginal vector, so recovery under this condition is independent of the tie-breaking rule.

An equivalent formulation uses a directed trellis whose state records the running sum in each check. Unit flows through this trellis represent probability distributions on satisfying local assignments. Equating their symbol marginals gives an LP with $O(mq^2)$ variables and constraints, solvable in polynomial time; see \cref{app:lp-reliability} for more details.

\subsubsection{A positive error-rate guarantee}

The following theorem bounds the LP decoder's failure probability.
\begin{theorem}[Finite-graph LP reliability]\label{thm:lp-reliability}
Let $G$ be a simple graph with $m$ edges, maximum degree at most $D\ge3$, and girth $g_G$. Consider any homological cycle code over $\mathbb F_q$ associated with $G$, with each position corrupted independently with probability $0<p<1/2$ and arbitrary nonzero error values. For every integer $1\le h<g_G$, the probability that the LP decoder declares failure or returns an incorrect word satisfies
\begin{equation}\label{eq:lp-finite-bound}
 P_{\mathrm W}^{\mathrm{LP}}(G,p)
 \le 2m(D-1)^{h-1}\bigl[2\sqrt{p(1-p)}\bigr]^h.
\end{equation}
The same bound holds when exactly $t\le p m$ positions are chosen uniformly without replacement and assigned arbitrary nonzero error values.
\end{theorem}

The proof in \cref{app:lp-reliability} shows that the true error is the unique LP optimum whenever every simple path of $h$ edges contains fewer than $h/2$ corrupted positions. Counting paths and applying a binomial or hypergeometric tail bound gives \cref{eq:lp-finite-bound}. Related path-based conditions for binary LP decoding appear in~\cite{EH12}. Requiring fewer than $h/2-1$ corrupted positions on each path also guarantees recovery after every one-coordinate change. The failure probability for this stronger condition is at most
\begin{equation}\label{eq:lp-robust-finite-bound}
 2m(D-1)^{h-1}\frac{1-p}{p}
       \bigl[2\sqrt{p(1-p)}\bigr]^h.
\end{equation}
The extra factor is independent of block length.

On the Linial--Simkin ensemble, the girth grows logarithmically with $n$. Choosing $h$ just below the girth in \cref{eq:lp-finite-bound} gives the following positive error-rate guarantee.
\begin{corollary}[Linear-error LP recovery on the Linial--Simkin ensemble]
\label{cor:lp-linial}
    Fix $D\ge3$, a finite field $\mathbb F_q$, and $c\in(0,1)$. Define the certified rate
    \begin{equation}\label{eq:lp-certified-rate}
     p_{\mathrm{LP}}^{\mathrm{cert}}(D,c)
     \coloneqq \frac12\left(1-\sqrt{1-(D-1)^{-2(1+1/c)}}\right)>0.
    \end{equation}
    For every fixed $0<p<p_{\mathrm{LP}}^{\mathrm{cert}}(D,c)$, there is a constant $\kappa>0$ such that
    \begin{equation}\label{eq:lp-linial-rate}
     \sup_{G\in\operatorname{supp}\mathcal L_{n,D,c}}
          P_{\mathrm W}^{\mathrm{LP}}(G,p)
     =O_{D,c,p}(n^{-\kappa})\rightarrow0.
    \end{equation}
    The same conclusion holds for a uniformly random support of exactly $\lfloor p m\rfloor$ nonzero errors, where $m=Dn/2$. Both conclusions also hold for the probability of failure on the error or any one-coordinate modification, uniformly over the nonzero error values.
\end{corollary}
\begin{proof}
    Put $R=2(D-1)\sqrt{p(1-p)}$. The strict inequality in the hypothesis is equivalent to $R<(D-1)^{-1/c}$. Every successful output $G$ satisfies $g_G\geq g(n)$ by \cref{eq:ml-girth-schedule}. For sufficiently large $n$, choose  $h = g(n)-1 < g_G$, so that  $h = c\log_{D-1} n+O(1)$, and apply the finite LP bound \cref{eq:lp-finite-bound}. Uniformly over these graphs,
    \begin{equation}
     P_{\mathrm W}^{\mathrm{LP}}(G,p)
     \le\frac{nD}{D-1}R^h
     =O_{D,c,p}\!\left(n^{1+c\log_{D-1} R}\right).
    \end{equation}
    Take $\kappa=-1-c\log_{D-1} R>0$. The fixed-weight statement follows from the same finite-graph bound, and the one-coordinate statement follows from \cref{eq:lp-robust-finite-bound}.
\end{proof}

The bound holds uniformly over every graph produced by the ensemble. The certified rate is independent of $q$, because the proof uses only the error locations and discards information about the nonzero error values. It can be smaller than the actual LP threshold and does not establish the factor-$(q-1)$ improvement proved for ML decoding. The guarantee applies to random error locations, but it does not imply recovery from every error pattern of linear weight. The numerical experiments in the next subsection compare the noise rates tolerated by LP decoding with the ML decoding threshold at accessible block lengths.

\subsubsection{Numerical experiments}
\label{sec:lp-experiments}
We tested the LP decoder on the Linial--Simkin ensemble with $c=0.9$, $D\in\{3,5,7\}$, and prime fields $\mathbb F_q$ with $q\in\{2,3,5\}$. For each degree and vertex count $n$, we sampled eight independent graphs from the LS ensemble (\cref{def:linial-simkin-ensemble}). Errors were drawn from the $q$SC (\cref{def:qsym}); at each noise level we used between one and eight error vectors per graph, for a total of 60,048 error samples. We solved an equivalent split-check formulation of \cref{eq:lp-local-polytope} using the dual-simplex solver in HiGHS~1.15.1. Integral outputs were checked against the syndrome using exact field arithmetic. Fractional optima were counted as decoding failures, as in the definition of $P_{\mathrm W}^{\mathrm{LP}}$.

\Cref{fig:lp-waterfalls} shows the resulting \emph{waterfall curves}, which plot the word-error rate as a function of the input symbol-error probability. Each point averages equally over the eight graphs.\footnote{The numerical choice among degenerate optima can depend on the solver state. Filled symbols use a reused solver object with its solver state cleared between calls; open symbols use a fresh object for each call, with all original samples at that point rerun. The latter points allow up to eight hours per solve. The 38 remaining unresolved samples occur at seven high-noise quinary points and are represented by intervals, never omitted or assigned a decoding outcome.} As expected, the transition sharpens with block length, most clearly for the cubic graphs. For $D=3,q=2$, the input error probabilities at 50\% word error are $0.0991,0.0919,0.0849,0.0761$ with increasing instance size. The corresponding 10\%--90\% transition widths decrease from $0.0476$ to $0.0292$, $0.0212$, and $0.0162$. These descriptive crossing estimates use monotone interpolation of the measured curves. Both the narrowing and the movement of the transition are appreciable, so the largest measured midpoint should not be identified with any asymptotic threshold.

\begin{figure}[p]
\centering
\includegraphics[width=0.325\textwidth]{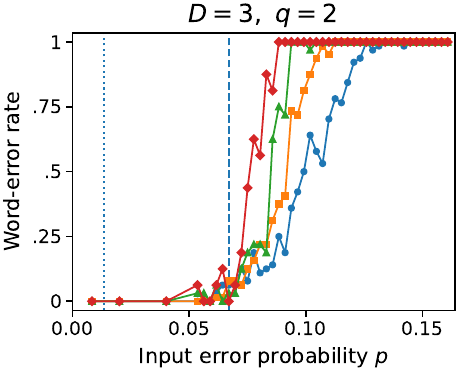}\hfill
\includegraphics[width=0.325\textwidth]{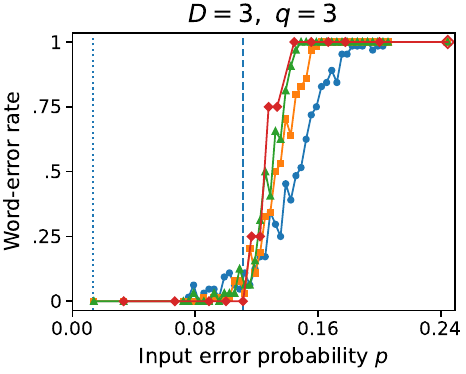}\hfill
\includegraphics[width=0.325\textwidth]{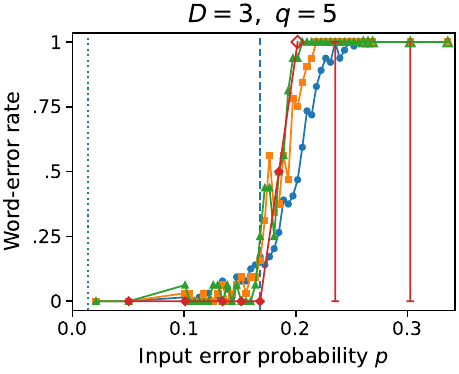}\\[3pt]
\includegraphics[width=0.325\textwidth]{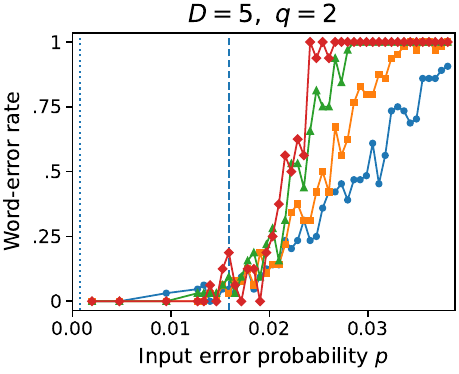}\hfill
\includegraphics[width=0.325\textwidth]{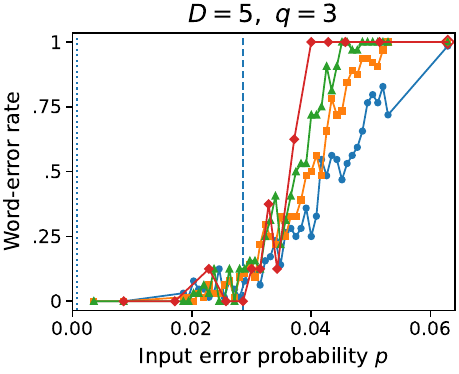}\hfill
\includegraphics[width=0.325\textwidth]{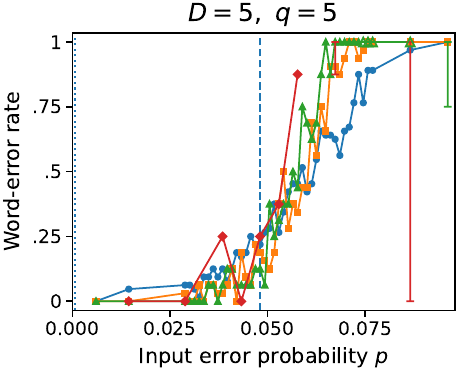}\\[3pt]
\includegraphics[width=0.325\textwidth]{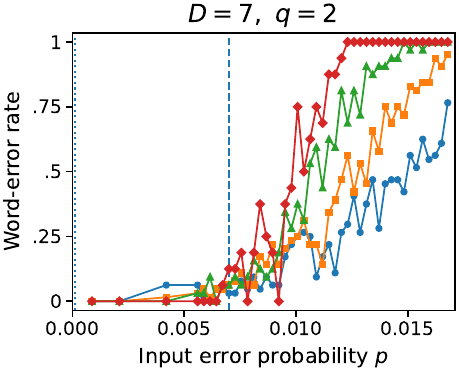}\hfill
\includegraphics[width=0.325\textwidth]{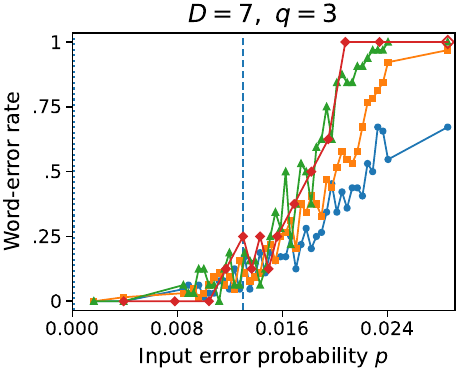}\hfill
\includegraphics[width=0.325\textwidth]{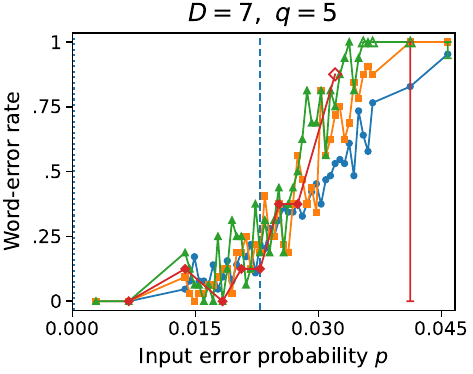}\\[3pt]
\includegraphics[width=\textwidth]{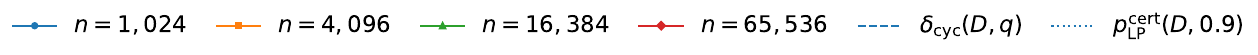}
\caption{LP word-error rates on Linial--Simkin graphs with $c=0.9$ and $m=Dn/2$ symbols. Each curve uses eight graphs. For $q=2,3,5$, respectively, the sample counts per noise level are $64,64,64$ at $n=1024$; $64,64,32$ at $n=4096$; $32,32,16$ at $n=16384$; and $16,8,8$ at $n=65536$. The dashed lines mark the ML cycle upper bound~\cref{eq:ml-cycle-value}; the dotted lines mark the certified LP rate~\cref{eq:lp-certified-rate}. Open symbols distinguish complete fresh-object solves from the standard runs. Vertical intervals span the possible full-sample word-error rates when some solves time out; they are not statistical error bars. Curves are not continued through such points.
}

\label{fig:lp-waterfalls}
\end{figure}

The observed useful decoding range extends well beyond the lower bound presented in \cref{cor:lp-linial}. For example, at $D=7,q=5$ the certified rate is approximately $1.30\times10^{-4}$, whereas all 56 errors tested at $p=0.00686$ across $n=4096,16384,65536$ were recovered. The eight largest-size samples contained between 1512 and 1620 erroneous symbols. The cycle value provides a second reference for the numerical curves, although the explicit small-alphabet lower bracket does not apply at these modest degrees. Finite-size transition points above the cycle value are consistent with an asymptotic strong converse.

\section{DQI performance and classical comparison}
\label{sec:performance-DQI}

In this concluding section, we combine the analysis of \cref{sec:dqi-imperfect-decoding} with the decoding results of \cref{sec:decoding-cycle-codes} to study DQI performance on problems dual, in the Regev sense, to decoding homological cycle codes. In particular, we focus on the Linial--Simkin ensemble and DQI instantiated by LP decoding in superposition. The following corollary provides the direct connection needed to make our performance claims below. 

\begin{corollary}[DQI from cycle-code reliability]
\label{cor:cycle-dqi-performance}
    Fix $D\ge3$, a finite field $\mathbb F_q$, and $c\in(0,1)$. Consider Max-2-LINSAT instances on the LS ensemble $G\sim\mathcal L_{n,D,c}$ (\Cref{def:linial-simkin-ensemble}) with target sets of common size $r$ such that the corresponding noise channels all have a shared error distribution $\theta$. Let $\rho=r/q$. Call a cutoff rate $p\in[0,1/2]$ reliable if there are syndrome decoders returning zero on the zero syndrome for which the failure probability in \cref{thm:fixed-target-semicircle}, with $\ell=\lfloor pm\rfloor$, tends to zero in ensemble mean. Let $F^{\mathrm{DQI}}_{D,\theta,r}$ be the supremum of the limiting expected satisfaction fractions obtained by optimizing the weights $\mathbf w$ at reliable cutoff rates. Then
    \begin{equation}\label{eq:cycle-dqi-upper}
        F^{\mathrm{DQI}}_{D,\theta,r}
        \le F_\rho\!\left(\min\{\delta_\mathrm{cyc}(D,\theta),1-\rho\}\right).
    \end{equation}
    If $\gamma=q/D^{2/3}\le1/32$, then also
    \begin{equation}\label{eq:cycle-dqi-lower}
        F^{\mathrm{DQI}}_{D,\theta,r} \ge F_\rho\!\left(\min\{(1-16\gamma)\delta_\mathrm{cyc}(D,\theta),1-\rho\}\right).
    \end{equation}
    For $q=2$, the upper bound is attained for every $D\ge3$.
\end{corollary}
\noindent The bounds concern cutoffs for which decoding is asymptotically reliable. If $\delta_\mathrm{cyc}(D,\theta)<1/2$, the upper bound also excludes reliable cutoff rates above $1/2$. This includes uniform replacement at every $D,q$.
\begin{proof}
    At a reliable rate $p$, the optimized fraction converges in probability to $F_\rho(\min\{p,1-\rho\})$ by \cref{thm:fixed-target-semicircle}. It also converges in mean, because satisfaction fractions lie in $[0,1]$. If $p>\delta_\mathrm{cyc}(D,\theta)$, choose $\delta_\mathrm{cyc}(D,\theta)<p'<\min\{p,1/2\}$. A binomial error count at rate $p'$ is at most $\lfloor pm\rfloor$ with probability tending to one. The assumed uniform shell recovery would therefore imply reliable channel decoding at $p'$, contradicting \cref{thm:ml-general-converse}. This proves \cref{eq:cycle-dqi-upper}.
    
    Conversely, \cref{thm:ml-small-alphabet,lem:robust-minweight-recovery} give one-coordinate recovery for the lexicographically first minimum-weight decoder at every $p'\le(1-16\gamma)\delta_\mathrm{cyc}(D,\theta)$. For any smaller $p$, \cref{lem:robust-channel-cutoff} transfers this guarantee to all shells through $\lfloor pm\rfloor$. Apply \cref{thm:fixed-target-semicircle} and let $p$ increase to the certified endpoint to obtain \cref{eq:cycle-dqi-lower}. For $q=2$, the binary recovery proof establishes the same property at every $p'<\delta_\mathrm{cyc}(D,2)$. Substitute \cref{eq:ml-binary-exact} into $F_{1/2}(p)=1/2+\sqrt{p(1-p)}$.
\end{proof}
Note that the analysis of \cref{sec:dqi-imperfect-decoding} requires a deterministic, syndrome-consistent output for every valid syndrome, whereas the LP decoder declares failure in certain scenarios. To make it a valid DQI decoder, we modify it as follows. When the LP declares failure, return a solution of $H\mathbf y=\mathbf s$ found by Gaussian elimination, with free coordinates set to zero in a fixed order. This defines a polynomial-time map $D_\mathrm{LP}$ with $D_\mathrm{LP}(0)=0$. The word-error probability of the resulting decoder is at most $P_{\mathrm W}^{\mathrm{LP}}$, since declared failures were already counted as errors, so it inherits \cref{eq:lp-linial-rate}.

For a fixed-target DQI instance with $\ell=\lfloor pm\rfloor$ and $p<p_{\mathrm{LP}}^{\mathrm{cert}}$, the fixed-weight version makes the failure probability in \cref{thm:fixed-target-semicircle} tend to zero, uniformly over weights through $\ell$. Consequently the optimized satisfaction fraction tends to $F_{r/q}(\min\{p,1-r/q\})$ for arbitrary target sets of size $r$.

\subsection{MaxCut and \texorpdfstring{Max-$k$-Cut}{Max-k-Cut}}\label{sec:max_k_cut}
Having established a firm connection between decoding LS cycle codes and DQI performance, we turn to our first application of interest: the well-known problem of MaxCut and its generalization, Max-$k$-Cut.
\begin{definition}[Max-$k$-Cut]\label{def:maxkcut}
    Let $G=(V,E)$ be a simple graph. Max-$k$-Cut asks for an assignment $\mathbf x:V\to[k]$ maximizing the number of edges $\{u,v\}\in E$ with $x_u\ne x_v$. The cut fraction is this number divided by $m=|E|$. MaxCut is Max-2-Cut.
\end{definition}
\noindent As the name suggests, MaxCut admits a visualization as the number of edges ``cut'' by a bipartition of the vertices. More generally, Max-$k$-Cut can be understood as a vertex coloring problem in which, given a fixed number of colors $k$, one seeks to color the vertices to minimize the number of edges whose endpoints have the same color.

Max-$k$-Cut is an instance of Max-2-LINSAT when $k=q$ is a prime power. We identify the labels with $\mathbb F_q$ and orient the edges. The constraints are $x_v-x_u\in\mathbb F_q\setminus\{0\}$, so $B^T$ is an oriented incidence matrix, $r=q-1$, and $\abs{\widetilde{g}_i(a)}^2 = 1/(q-1)$ for every $i$ and $a\neq 0$. The semicircle formula is $b_q(p)\coloneqq F_{(q-1)/q}(p)$, increasing from the random-assignment fraction $(q-1)/q$ to $1$ on $0\leq p\leq 1/q$.

Write $F_{q,D}^{\mathrm{DQI}}$ for the optimal asymptotic cut fraction among the reliably decoded cutoff families of \cref{cor:cycle-dqi-performance} on the Linial--Simkin ensemble. For this Max-$k$-Cut comparison, take $q<D$, as in \cref{tab:tpm-dqi}. Here $\delta_\mathrm{cyc}(D,q)<1/(D-1)\le1/q$. Substituting the ML decoding threshold bound \cref{eq:ml-cycle-value} into the semicircle formula gives
\begin{equation}\label{eq:maxqcut-reliable-optimum}
 F_{q,D}^{\mathrm{DQI}}\le U_q(D)
 \coloneqq b_q\bigl(\delta_\mathrm{cyc}(D,q)\bigr).
\end{equation}
For $q/D^{2/3}\le1/32$, the same corollary gives the lower bound $b_q\bigl((1-16q/D^{2/3})\delta_\mathrm{cyc}(D,q)\bigr)$. Consequently, for $q=o(D^{2/3})$,
\begin{equation}\label{eq:DQI-MaxqCut-asymptotic-fraction}
 F_{q,D}^{\mathrm{DQI}}
 =\frac{q-1}{q}\left(1+\frac1{D-1}\right)+o(D^{-1}).
\end{equation}
For MaxCut, the exact binary threshold in \cref{eq:ml-binary-exact} gives the exact reliable-decoding value
\begin{equation}\label{eq:DQI-MaxCut-fraction}
 F_{2,D}^{\mathrm{DQI}}
 =\frac12+\sqrt{\delta_\mathrm{cyc}(D,2)
                         (1-\delta_\mathrm{cyc}(D,2))}
 =\frac12+\frac1{2(D-1)}.
\end{equation}
The polynomial-time LP decoder attains cut fractions approaching $b_q\bigl(p_{\mathrm{LP}}^{\mathrm{cert}}(D,c)\bigr)$ by \cref{cor:lp-linial,thm:fixed-target-semicircle}. For $q>2$, the sharper lower bound above uses minimum-weight decoding, which is NP-hard in general.

We first compare DQI performance to TPM. We observe that TPM gives a larger improvement at large degree on the same ensemble. As the girth grows, the inner product of adjacent unit vectors tends to $\sigma_D=-2\sqrt{D-1}/D$. Assigning each vertex the label with largest inner product among $q$ independent standard Gaussian vectors gives the limiting expected cut fraction
\begin{equation}\label{eq:tpm-limit}
 T_q(D)=1-q\int_{\mathbb R^2}
       \phi_{\sigma_D}(x,y)\Phi_{\sigma_D}(x,y)^{q-1}\,dx\,dy,
 \qquad \sigma_D=-\frac{2\sqrt{D-1}}D,
\end{equation}
where $\phi_\sigma$ and $\Phi_\sigma$ are the standard bivariate Gaussian density and distribution function with correlation $\sigma$. The integral is the probability that one specified label wins at both endpoints. For $q=2$, $T_2(D)=\arccos(\sigma_D)/\pi$. Expanding at zero correlation gives, with independent $Z_a\sim\mathcal N(0,1)$,
\begin{equation}\label{eq:tpm-large-degree}
 T_q(D)=\frac{q-1}{q}+\frac{\kappa_q}{\sqrt D}+O_q(D^{-1}),
 \qquad \kappa_q=\frac{2}{q-1}\left(\mathbb E\max_{1\le a\le q}Z_a\right)^2>0.
\end{equation}
This exceeds the $O_q(D^{-1})$ improvement of reliably decoded DQI at large $D$. The TPM limit also exceeds $U_q(D)$ at every parameter pair in \cref{tab:tpm-dqi}.

Second, we compare DQI with simulated annealing (SA). We start each run from a uniformly random assignment and return the best assignment visited. A sweep visits every vertex, proposing a uniformly chosen different label and accepting it by the Metropolis rule, with energy equal to the number of unsatisfied constraints. The inverse temperature increases linearly between sweeps, with endpoints chosen on separate tuning instances. For Max-$k$-Cut, the same evaluation graphs are used for every field at a given degree; their sizes, girth, and sampling errors are given in \cref{tab:tpm-dqi}.

\begin{table}[H]
\centering
\footnotesize
\setlength{\tabcolsep}{4pt}
\renewcommand{\arraystretch}{1.12}
\captionsetup{font=footnotesize,skip=4pt}
\begin{tabular*}{\textwidth}{@{\extracolsep{\fill}}clrrrrr@{}}
\toprule
$D$ & Method & $q=2$ & $q=3$ & $q=5$ & $q=7$ & $q=11$\\
\midrule
 & Random & 0.50000 & 0.66667 & 0.80000 & 0.85714 & 0.90909\\
\midrule
3 & DQI & 0.75000 & \textendash & \textendash & \textendash & \textendash\\
 & TPM & 0.89183 & \textendash & \textendash & \textendash & \textendash\\
 & SA & 0.92243 & \textendash & \textendash & \textendash & \textendash\\
\addlinespace[3pt]
4 & DQI & 0.66667 & 0.85553 & \textendash & \textendash & \textendash\\
 & TPM & 0.83333 & 0.95549 & \textendash & \textendash & \textendash\\
 & SA & 0.86671 & 1.00000 & \textendash & \textendash & \textendash\\
\addlinespace[3pt]
5 & DQI & 0.62500 & 0.81427 & \textendash & \textendash & \textendash\\
 & TPM & 0.79517 & 0.93381 & \textendash & \textendash & \textendash\\
 & SA & 0.83374 & 0.99893 & \textendash & \textendash & \textendash\\
\addlinespace[3pt]
6 & DQI & 0.60000 & 0.78764 & 0.92160 & \textendash & \textendash\\
 & TPM & 0.76772 & 0.91595 & 0.98105 & \textendash & \textendash\\
 & SA & 0.80354 & 0.98571 & 1.00000 & \textendash & \textendash\\
\addlinespace[3pt]
7 & DQI & 0.58333 & 0.76911 & 0.90588 & \textendash & \textendash\\
 & TPM & 0.74675 & 0.90112 & 0.97455 & \textendash & \textendash\\
 & SA & 0.78393 & 0.96962 & 1.00000 & \textendash & \textendash\\
\addlinespace[3pt]
13 & DQI & 0.54167 & 0.71999 & 0.85919 & 0.91607 & 0.96358\\
 & TPM & 0.67891 & 0.84677 & 0.94473 & 0.97328 & 0.99060\\
 & SA & 0.70968 & 0.90178 & 0.99950 & 1.00000 & 1.00000\\
\addlinespace[3pt]
17 & DQI & 0.53125 & 0.70707 & 0.84569 & 0.90341 & 0.95320\\
 & TPM & 0.65596 & 0.82638 & 0.93148 & 0.96439 & 0.98596\\
 & SA & 0.68357 & 0.87517 & 0.98809 & 1.00000 & 1.00000\\
\bottomrule
\end{tabular*}
\caption{\label{tab:tpm-dqi}Fractions of edges cut for Max-$k$-Cut. DQI is the upper bound $b_q\bigl(\delta_\mathrm{cyc}(D,q)\bigr)$ for the optimized reliable-cutoff performance in \cref{eq:maxqcut-reliable-optimum}; TPM is the infinite-girth limit in \cref{eq:tpm-limit}. SA is the mean best cut fraction from one 65,536-sweep anneal on each of 16 independent Linial--Simkin graphs with $n=65536$ and girth at least $\max\{5,\lfloor0.9\log_{D-1}n\rfloor\}$. Temperature endpoints were selected on separate tuning graphs. All SA standard errors are below $3.4\times10^{-5}$; every SA entry of $1.00000$ represents perfect cuts on all 16 graphs. Values are rounded to five decimals, and dashes mark the excluded cases $q\ge D$.}
\end{table}

\Cref{tab:tpm-dqi} presents a comparison of all three methods. SA gives the largest mean cut fraction in all 21 cases, exceeding the TPM limit by $0.94$--$6.98$ percentage points and the DQI upper bound by $3.64$--$20.87$ points. Seven cases yield perfect cuts on all 16 graphs. The theoretical entries are asymptotic; the SA entries are finite-size measurements. On four evaluation graphs per parameter pair, the SA cut fractions after only 1,024 sweeps exceed both theoretical references. Increasing the budget from 65,536 to 262,144 sweeps improves the nonperfect means by less than $0.095$ percentage points. At $n=16384$, using eight additional graphs per degree and 65,536 sweeps, the means differ from those at $n=65536$ by at most $0.071$ percentage points.

\subsection{Random Max-2-LINSAT} \label{subsec:random-max2LINSAT}
For other Max-2-LINSAT objectives, we first chose independent uniform singleton targets $C_i=\{v_i\}$, so an oriented edge is satisfied when $x_v-x_u=v_i$. Here $r=1$ and $\theta(a)=1/(q-1)$ for every target vector. Substituting \cref{eq:ml-Z,eq:ml-cycle-value} into the semicircle formula gives the reliable-decoding upper bound
\begin{equation}\label{eq:numerical-linsat-reference}
 F_{1/q}(\delta_\mathrm{cyc}(D,q))
 =\frac{1+(q-1)Z_q(\delta_\mathrm{cyc}(D,q))}{q}
 =\frac1q+\frac{q-1}{q(D-1)}.
\end{equation}
For $D\in\{3,5,7\}$ and $q\in\{2,3,5\}$, SA exceeds this bound on every tested instance at each budget of 16, 1,024, and 65,536 sweeps. Each instance has one anneal per budget, with its targets held fixed. At $n=65536$, increasing the budget from 1,024 to 65,536 sweeps improves the mean satisfaction fraction by at most $2.74$ percentage points across these nine singleton-target cases, with the largest increase at $(q,r,D)=(5,1,3)$. At 65,536 sweeps, the means from eight instances at $n=4096$ differ from those at $n=65536$ by less than $0.08$ percentage points.

We also chose each target independently and uniformly among the $r$-element subsets of $\mathbb F_q$, for $q=3,5$ and $r=2,\ldots,q-1$. \Cref{tab:sa-lists} collects the 65,536-sweep SA results for all studied target sizes, alongside the DQI values $F_{r/q}(\min\{\delta_\mathrm{cyc}(D,q),1-r/q\})$ from \cref{cor:DQI-error-random-binomial-weights} under the decoder assumption in the caption. SA exceeds these DQI values in every case.

\begin{table}[H]
\centering
\begin{tabular}{rr rr rr rr}
\toprule
 & & \multicolumn{2}{c}{$D=3$} & \multicolumn{2}{c}{$D=5$} & \multicolumn{2}{c}{$D=7$}\\
\cmidrule(lr){3-4}\cmidrule(lr){5-6}\cmidrule(lr){7-8}
$q$ & $r$ & SA & DQI & SA & DQI & SA & DQI\\
\midrule
2 & 1 & 0.9223 & 0.7500 & 0.8337 & 0.6250 & 0.7840 & 0.5833 \\
3 & 1 & 0.8870 & 0.6667 & 0.7636 & 0.5000 & 0.6969 & 0.4444 \\
3 & 2 & 1.0000 & 0.9259 & 0.9990 & 0.8143 & 0.9696 & 0.7691 \\
5 & 1 & 0.8497 & 0.6000 & 0.6968 & 0.4000 & 0.6148 & 0.3333 \\
5 & 2 & 0.9999 & 0.8000 & 0.9159 & 0.6192 & 0.8438 & 0.5511 \\
5 & 3 & 1.0000 & 0.9328 & 1.0000 & 0.8000 & 0.9812 & 0.7419 \\
5 & 4 & 1.0000 & 0.9983 & 1.0000 & 0.9423 & 1.0000 & 0.9059 \\
\bottomrule
\end{tabular}
\caption{Mean satisfied fractions for random target sets of size $r$ at $n=65536$ and $c=0.9$. SA uses one 65,536-sweep anneal per instance, with eight independent graph/target instances for $r=1$ and four for $r>1$; all standard errors are below $1.4\times10^{-4}$. DQI gives the asymptotic mean satisfaction fraction with optimized binomial amplitudes, assuming a target-independent decoder reliable up to the ML cycle bound. For intermediate target sizes, these values do not bound decoders that exploit the individual target sets. Values are rounded to four decimal places; a displayed $1.0000$ can therefore include a small number of unsatisfied constraints.}
\label{tab:sa-lists}
\end{table}

\paragraph{Acknowledgments} We thank Noah Shutty, Gilles Z{\'e}mor, and Earl Campbell for helpful discussions. The JPMorganChase team thanks Rob Otter for the executive support of the work and all of their colleagues at the Global Technology Applied Research center for helpful feedback and support throughout the project.

O.P.\ was supported by the U.S.\ Department of Energy, Office of Science, Accelerated Research in Quantum Computing, Fundamental Algorithmic Research toward Quantum Utility (FAR-Qu). This article has been authored by an employee of National Technology \& Engineering Solutions of Sandia, LLC under Contract No.\ DE-NA0003525 with the U.S. Department of Energy (DOE). The employee owns all right, title and interest in and to the article and is solely responsible for its contents. The United States Government retains and the publisher, by accepting the article for publication, acknowledges that the United States Government retains a non-exclusive, paid-up, irrevocable, world-wide license to publish or reproduce the published form of this article or allow others to do so, for United States Government purposes. The DOE will provide public access to these results of federally sponsored research in accordance with the DOE Public Access Plan \url{https://www.energy.gov/downloads/doe-public-access-plan}.

\paragraph{AI statement} Various AI models were used throughout this project, but all work done by AI was supervised and carefully checked by the authors, who take responsibility for the correctness of all claims.

\section*{Disclaimer}
This paper was prepared for informational purposes by the Global Technology Applied Research center of JPMorgan Chase \& Co. This paper is not a product of the Research Department of JPMorgan Chase \& Co. or its affiliates. Neither JPMorgan Chase \& Co. nor any of its affiliates makes any explicit or implied representation or warranty and none of them accept any liability in connection with this paper, including, without limitation, with respect to the completeness, accuracy, or reliability of the information contained herein and the potential legal, compliance, tax, or accounting effects thereof. This document is not intended as investment research or investment advice, or as a recommendation, offer, or solicitation for the purchase or sale of any security, financial instrument, financial product or service, or to be used in any way for evaluating the merits of participating in any transaction.

\begin{appendices}
\addtocontents{toc}{\protect\setcounter{tocdepth}{1}} 

\section{Proofs of the Maximum-Likelihood bounds}
\label[appendix]{app:ml-small-alphabet}
We prove \cref{thm:ml-general-converse,thm:ml-small-alphabet} for the Linial--Simkin ensemble. Fix $D\ge3$, $q$, $c\in(0,1)$, and the channel during the $n$-limit. Throughout this appendix, abbreviate the prescribed minimum girth in Eq.~\eqref{eq:ml-girth-schedule} by $g\coloneqq g(n)$. 

Let $\mathcal E_n$ be the event of successful completion of all $D-2$ matching rounds in \cref{def:linial-simkin-ensemble}. Thus $\Pr(\mathcal E_n)=1-o(1)$ by~\cite{LS}. We may randomize the initial Hamilton cycle $H$ and its direction: this randomizes labels without changing graph-invariant decoding probabilities. Unless explicitly conditioned, probabilities refer to the unconditioned process. A failed run contributes zero to a count requiring a completed graph; conditioning on $\mathcal E_n$ is imposed at the end. The published graph-construction and process estimates are cited below; the matching, planting, and coding arguments are additional arguments of this manuscript.

\subsection{Dense matching rounds and prescribed subgraphs}\label{ls:sec:subgraphs}

For a given Hamilton cycle $H$ on $V=[n]$, with $n$ even, and with the target degree $D\in\mathbb{N}$ and target girth $g$ fixed above, the Linial--Simkin $(H,g,D)$\textit{-high-girth process} is a graph sequence $(G_0, G_1, \ldots)$ that begins with $G_0=H$ and where $G_{t+1}$ is constructed from $G_t$ as follows: If $G_t=(V, E_t)$ is $D$-regular, then $G_{t+1}=G_t$. Otherwise, let $H_t=(W_t, \mathscr A_t)$ where
\begin{equation}
    W_t \coloneqq \argmin_{v\in V} d_{G_t}(v)
\end{equation}
is the set of \textit{unsaturated} vertices and
\begin{equation}
    \mathscr A_t \coloneqq \{\{u,v\}\subset W_t\,|\,\delta_{G_t}(u,v) \geq g-1\}
\end{equation}
is the set of \textit{available} edges. If $\mathscr A_t = \emptyset$, then $G_{t+1}=G_t$. Otherwise, $E_{t+1} = E_t \cup \{e_{t+1}\}$ where $e_{t+1}$ is chosen uniformly at random from $\mathscr A_t$. The process \textit{freezes} at time $t$ if $G_t=G_{t+1}$ and \textit{saturates} at time $t$ if $G_t$ is $D$-regular. For the logarithmic target girth specified in Eq.~\eqref{eq:ml-girth-schedule}, the process saturates with high probability by~\cite[Theorem 1.1]{LS}. It proceeds in up to $D-2$ successive \textit{rounds}, where round $k \in \{3, \dots, D\}$ corresponds to adding a random matching $M_{k-2}$ appearing in \eqref{eq:ls-decomposition}. If completed, this round takes the process from a $(k-1)$-regular graph to a $k$-regular graph.

We require the following uniform lower bound on the degree of vertices in $H_t$. Its proof relies on various intermediate results in~\cite{LS}. For convenience in this lemma, we adopt a restart convention where $t = 0$ at the beginning of each round, and write $G_t$ for its evolving graph. In this ``round-local" notation, 
\begin{equation}
    W_t = \{v \in V : d_{G_t}(v) = k-1\}
\end{equation}
and $\abs{W_t} = n - 2t$ until the round freezes or completes.
\begin{lemma}[Uniform density along a round]\label[lemma]{lem:density}
    Consider a round from degree $k-1$ to degree $k$, where $3 \le k \le D$. For any $a>0$, let $\mathcal A_n^{(k)}$ denote the event that, at every time $t$ during round $k$,
    \begin{equation}\label{ls:eq:density}
        d_{\mathscr A_t}(v)\ge(1-\zeta_n)(|W_t|-1)
        \quad(v\in W_t),\qquad \zeta_n \coloneqq n^{-a}.
    \end{equation}
    Then there exists a constant $a=a(D,c)>0$ such that $\Pr[\mathcal A_n^{(k)}] = 1 - o(1)$ uniformly over all $(k-1)$-regular graphs of girth at least $g$ at the start of the round. Moreover, $\mathcal A_n^{(k)}$ implies saturation.
\end{lemma}
\begin{remark}
    The proof of the lemma mirrors the general strategy employed in the proof of~\cite[Proposition 2.1]{LS}. It splits the round into intervals using a finite sequence of moments $t_0, t_1,\ldots, t_J$ engineered so that $|W_{t_{j+1}}| = |W_{t_j}|n^{-\alpha} - O(1)$.
    This is achieved by exploiting the fact that if the process is not frozen, then $|W_t|=n-2t$. To prove the process does not freeze before the end of the round, the proof uses induction and the notion of $C$-\textit{path-boundedness}, introduced in~\cite[Definition 2.10]{LS}. 
\end{remark}
\begin{proof}
    Set
    \begin{equation}
        c_0=\frac{1+c}{2},\qquad
        \epsilon=\frac{c_0(1-c_0)}3,\qquad
        \alpha=\frac{\epsilon}{1000},\qquad
        a=\frac{\epsilon}{3},
    \end{equation}
    so that $g\le c_0\log_{k-1}n$ and $c_0>1/2$. Choose boundary times $t_0,\ldots, t_J$ that divide the round into an initial interval, $J$ intervening intervals, and a terminal interval, where $J = \lceil c_0/\alpha\rceil = O(1)$. 
    \begin{equation}
        t_0 = \left\lceil\frac12 \left(n - n^{c_0+\epsilon}\right)\right\rceil,\qquad
        t_j = \left\lceil\frac12 \left(n - (n - 2 t_{j-1}) n^{-\alpha}\right)\right\rceil
    \end{equation}
    We will prove that \eqref{ls:eq:density} holds w.h.p. for all times in each interval.
    
    Fix $t\in\{0,\ldots, t_0\}$ in the first interval. A pair $\{u,v\}\subset W_t$ is available in $\mathscr A_t$ if and only if $\delta_{G_t}(u, v) \geq g-1$. But $G_t$ has maximum vertex degree $k$, so the number of forbidden partners for $v\in W_t$ satisfies the Moore bound, \cite[Observation 2.2]{LS}
    \begin{equation}
        |W_t\setminus\mathscr A_t(v)| \leq \sum_{\ell=0}^{g-3} k(k-1)^{\ell} = O_k((k-1)^{g-2}) = O_k(n^{c_0})
    \end{equation}
    where $\mathscr A_t(v)$ is the available-neighbor set for vertex $v$. By~\cite[Lemma 2.3]{LS}, we have $|W_{t_0}|=n-2t_0=n^{c_0+\epsilon}-O(1)$, so
    \begin{equation}
        \frac{|W_t\setminus\mathscr A_t(v)|}{|W_t|-1} \leq \frac{O(n^{c_0})}{n^{c_0+\epsilon} - O(1)} = O(n^{-\epsilon}) \leq \zeta_n
    \end{equation}
    and thus
    \begin{equation}
        d_{\mathscr A_t}(v) \geq (1 - \zeta_n)\left(|W_t| - 1\right)
    \end{equation}
    for all $t\in\{0,\ldots, t_0\}$. Moreover, by~\cite[Lemma 2.9]{LS}, $G_{t_0}$ is $C_0$-path-bounded with $C_0=3$ w.h.p.
    
    Fix $j\in\{0,\ldots,J-1\}$ and assume $G_{t_j}$ is $C_j$-path-bounded. We use the auxiliary construction from the proof of~\cite[Lemma 2.11]{LS}. Form a random graph $R_j$ on $W_{t_j}$ by including each available edge independently with probability $n^{\epsilon/10}/\abs{W_{t_j}}$. Starting from $G_{t_j}$, run the high-girth process restricted to edges of $R_j$, and denote its successive graphs by $\widetilde{G}_t$. LS couples this restricted process to the original process up to the restricted process' freezing time. The analysis in the proof of~\cite[Claim 2.13]{LS} shows that, with high probability, the restricted process does not freeze before $t_{j+1}$. On this event,
    \begin{equation}
        G_t = \widetilde{G}_t \subseteq G_{t_j} \cup R_j \qquad (t_j \leq t \leq t_{j+1}).
    \end{equation}
    Call distinct vertices $u,v \in W_{t_j}$ $\ell$\emph{-threatened} if there is a sequence $w_0,\ldots, w_{2r-1} \in W_{t_j}$, with $w_0 = u$ and $w_{2r-1} = v$, such that
    \begin{equation}
        \{w_{2i-1}, w_{2i}\} \in E(R_j) \qquad (1\leq i \leq r-1)
    \end{equation}
    and 
    \begin{equation}
        (r-1) + \sum_{i=0}^{r-1} \delta_{G_{t_j}}(w_{2i}, w_{2i+1}) = \ell.
    \end{equation}
    Such a sequence witnesses a connection of total length $\ell$ alternating between paths in $G_{t_j}$ and candidate edges in $R_j$. Let $T_\ell(v)$ be the number of $\ell$-threatened pairs containing $v$, for this interval.
    
    On the successful coupling event, every forbidden pair $\{u,v\} \subset W_t$ has a path of length at most $g-2$ in $G_t \subseteq G_{t_j} \cup R_j$. Decomposing that path at its edges outside $G_{t_j}$ shows that the pair is $\ell$-threatened for some $1 \leq \ell \leq g-2$. Consequently,~\cite[Claim 2.12(c)]{LS} gives, with high probability and simultaneously for all $t_j \leq t \leq t_{j+1}$ and $v \in W_t$,
    \begin{equation}\label{eq:bound_using_L_function}
        |W_t \setminus \mathscr A_t(v)| \le 1 + \sum_{\ell=1}^{g-2} T_\ell(v) \le 1 + \log^{O(1)}(n) \sum_{\ell=1}^{g-2} L(\ell, t_j),
    \end{equation}
    where
    \begin{equation}\label{eq:def_L_function}
        L(\ell, t) = \max\left\{1, \frac{(k-1)^\ell (n - 2t)}{n^{c_0+\epsilon}}\right\} \leq 1 + \frac{(k-1)^\ell (n - 2t)}{n^{c_0+\epsilon}}.
    \end{equation}
    Using $\sum_{\ell=1}^{g-2} (k-1)^\ell = O((k-1)^g) = O(n^{c_0})$ and $n - 2t_j = |W_{t_j}|$, the bound \eqref{eq:bound_using_L_function} becomes
    \begin{equation}\label{eq:bound_for_intermediate_intervals}
        |W_t \setminus \mathscr A_t(v)| \leq \left(1 + \frac{|W_{t_j}|}{n^\epsilon}\right)\log^{O(1)}(n).
    \end{equation}
    Since the process does not freeze before $t_{j+1}$ on this event,
    \begin{equation}
        \abs{W_t} \geq \abs{W_{t_{j+1}}} = \abs{W_{t_j}} n^{-\alpha} - O(1).
    \end{equation}
    Combining this with \eqref{eq:bound_for_intermediate_intervals} gives
    \begin{equation}
        \frac{|W_t \setminus \mathscr A_t(v)|}{|W_t|-1} \le \frac{\left(1 + \frac{|W_{t_j}|}{n^\epsilon}\right)\log^{O(1)}(n)}{|W_{t_j}|n^{-\alpha} - O(1)} \le 2n^{-(\epsilon-\alpha)}\log^{O(1)}(n) \le n^{-(\epsilon-2\alpha)} \le \zeta_n
    \end{equation}
    which in turn implies that
    \begin{equation}
        d_{\mathscr A_t}(v) \geq (1 - \zeta_n)\left(|W_t| - 1\right)
    \end{equation}
    for all $t\in\{t_j,\ldots, t_{j+1}\}$ w.h.p. Moreover, by~\cite[Lemma 2.11]{LS}, $G_{t_{j+1}}$ is $C_{j+1}$-path-bounded w.h.p. for some constant $C_{j+1}$.
    
    Finally, let $B_t$ be the forbidden-pair set, and $P_\ell(G_t)$ be the number of unordered, unsaturated pairs at distance $\ell$. By~\cite[Definition 2.10(b)]{LS} and~\cite[Lemma 2.11]{LS} we have
    \begin{equation}
        |B_{t_J}| \leq \sum_{\ell=1}^{g-2} P_\ell(G_{t_J}) \leq |W_{t_J}|^2n^{c_0-1}\log^{O(1)}(n).
    \end{equation}
    However, $|W_{t_J}| \leq n^\epsilon$, so $|B_{t_J}| = o(1)$. Since $|B_{t_J}|$ is a non-negative integer, $B_{t_J}=\emptyset$ w.h.p., implying that $G_{t_J}$ is safe in the sense of~\cite[Definition 2.4]{LS}. By~\cite[Lemma 2.5]{LS}, $G_t$ is safe for all $t > t_J$ and the process saturates. In particular, for any $t\in\{t_J,\ldots,n/2\}$ the graph $H_t$ is complete w.h.p., which implies \eqref{ls:eq:density}.
    
    We have shown that \eqref{ls:eq:density} holds w.h.p. for each of the $J+2=O(1)$ intervals. By a union bound it holds simultaneously for every time $t$ in the round and for every vertex $v\in W_t$ w.h.p.
\end{proof}

Extending this single-round result to ensure high probability over all rounds is straightforward. Let
\begin{equation}
    \mathcal{A}_n \coloneqq \bigcap_{k=3}^D \mathcal{A}_n^{(k)}
\end{equation}
be the event that the density inequality~\eqref{ls:eq:density} holds over all rounds. By the preceding lemma, and a union bound, $\Pr(\mathcal{A}_n) = 1 - o(1)$. Moreover, $\mathcal{A}_n$ implies successful completion of the construction.

\subsubsection{The matching-completion potential}
For even $v$ and $0\le k\le v/2$, set
\begin{equation}\label{ls:eq:phi}
 \Phi_k(v)\coloneqq\prod_{j=0}^{k-1}(v-2j-1)^{-1},\qquad \Phi_0(v)=1.
\end{equation}
This is the probability that a uniform perfect matching on $v$ vertices contains a specified matching of size $k$.

\begin{lemma}[Prescribed matching]\label[lemma]{lem:matching}
Consider a round on $n$ vertices in which the next edge is uniform among available pairs. For a specified matching $F$ of size $s\le n/2$,
\begin{equation}\label{ls:eq:matching}
 \Pr(F\subseteq M,\ \cref{ls:eq:density}\text{ holds throughout}\mid\text{past})
 \le \exp(C\zeta_n s\log n)\Phi_s(n),
\end{equation}
where $C$ is an absolute constant. The assertion is uniform in the starting graph and in $F$.
\end{lemma}
\begin{proof}
As long as no endpoint of a pending edge of $F$ has been matched incorrectly, assign the state with $v$ free vertices and $k$ pending edges the potential $\Phi_k(v)$. Otherwise assign it zero. Also set it to zero when the density condition first fails. This is a killed process, not a process conditioned on a future event.

Put $N=\binom v2$. Suppose $M_0$ pairs are missing, including $b$ pairs wholly outside the $2k$ protected endpoints and $a_0$ prescribed pairs. For $v>2k$, the expected next potential divided by the current one, before any additional killing, equals
\begin{equation}\label{ls:eq:driftratio}
 \frac{N-a_0(v-1)-b(v-1)/(v-2k-1)}{N-M_0}.
\end{equation}
Indeed, selecting a prescribed edge multiplies the potential by $v-1$, selecting an outside edge multiplies it by $(v-1)/(v-2k-1)$, and every other choice kills it. For the complete graph these contributions sum to $N$.

The excess of \cref{ls:eq:driftratio} above one is at most
\begin{equation}
 \frac{M_0-b}{N-M_0}
 \le\frac{2k\zeta_n(v-1)}{(1-\zeta_n)\binom v2}
 \le\frac{8\zeta_n k}{v}
\end{equation}
for $\zeta_n\le1/2$. When $v=2k$, the ratio is at most $(1-\zeta_n)^{-1}$, which obeys the same final bound. Killing can only decrease the expectation.

Now $k\le s$ and $\sum_{v=n,n-2,\ldots,2}v^{-1}=O(\log n)$. Iterating the one-step bound gives \cref{ls:eq:matching}, since the terminal potential is the indicator of the stated event.
\end{proof}

\subsubsection{Port labels and small subgraphs}
Use the $D$ symbols
\begin{equation}
 \mathcal P=\{H^+,H^-,M_1,\ldots,M_{D-2}\},
\end{equation}
with involution $\overline{H^+}=H^-$ and $\overline{M_i}=M_i$. Traversing an edge with symbol $a$ uses port $a$ at its initial vertex and port $\bar a$ at its terminal vertex. A labeling is \emph{admissible} if incident ports are distinct. This records which round supplied each edge, and the direction of Hamilton edges.

\begin{lemma}[Prescribed labeled subgraph]\label[lemma]{lem:subgraph}
There is $\xi_n=n^{-\Omega(1)}\log n$ such that every specified admissibly port-labeled simple graph $J$ on specified vertices, with $w\le n/4$ edges, satisfies
\begin{equation}\label{ls:eq:subgraph}
 \Pr(J\subseteq G_n,\mathcal A_n)
 \le\frac{e^{\xi_n w}}{(n-2w)^w}.
\end{equation}
Consequently, uniformly for $w\le K_0\log n$, for any fixed $K_0$,
\begin{equation}\label{ls:eq:smallsubgraph}
 \Pr(J\subseteq G_n,\mathcal A_n)\le(1+o(1))n^{-w}.
\end{equation}
A connected unlabeled $w$-edge graph has at most $D(D-1)^{w-1}$ admissible port-labelings.
\end{lemma}
\begin{proof}
The prescribed Hamilton arcs either are impossible or form vertex-disjoint directed paths. If their number is $b<n$, contracting them gives probability
\begin{equation}
 \frac{(n-b-1)!}{(n-1)!}\le(n-b)^{-b}.
\end{equation}
In each matching round the prescribed edges form a matching, or the probability is zero. Apply \cref{lem:matching} successively, conditional on the preceding rounds. Each denominator in $\Phi_s(n)$ is at least $n-2w$, and the sum of all prescribed edge counts is $w$. This proves \cref{ls:eq:subgraph}. The stated rate for $\xi_n$ gives \cref{ls:eq:smallsubgraph}.

For the last assertion, order the edges so that each edge after the first meets a preceding one. The first edge has $D$ possible symbols. Each subsequent edge has at least one forbidden port at a previously encountered vertex, and therefore at most $D-1$ choices.
\end{proof}

\subsection{Recovery with the full alphabet dependence}\label{ls:sec:recovery}
For recovery, each coordinate may have its own conditional nonzero-error distribution $\theta_e$, with the same corruption probability $p$. Let $\Theta\in(0,1]$ be a fixed upper bound on $\max_{e,a\ne0}\theta_e(a)$ along the graph-size sequence. The common channel in \cref{eq:ml-general-noise} has $\Theta=\theta_\mathrm{max}$. Evenness is not needed for the recovery criterion or its one-coordinate extension. A nonzero flow $f\in C^\perp$ is \emph{Hamming-favorable} if $\wt{E-f}\le\wt E$. If a Hamming-favorable flow is disconnected, at least one connected component of its support is Hamming-favorable. It is therefore enough to consider connected supports. We prove recovery by minimum Hamming weight, which also bounds the error probability of the optimal ML decoder.

\subsubsection{Pairwise comparison and a budget on the support excess}
Write
\begin{equation}\label{eq:ml-hamming-affinity}
 B_\Theta(p)=2\sqrt{(1-p)p\Theta}+p(1-\Theta).
\end{equation}
For a common even distribution $\theta$ and $0\le p\le1/2$, Cauchy--Schwarz on the terms with $z\notin\{0,a\}$ in \cref{eq:ml-general-Z-expanded} gives
\begin{equation}\label{eq:ml-affinity-comparison}
 2\sqrt{(1-p)p\theta_\mathrm{max}}
 \le Z_\theta(p)\le B_{\theta_\mathrm{max}}(p)\le Z_\theta(p)+p.
\end{equation}
Indeed, the remainder for a shift $a$ is at most $p(1-\theta(a))$, and $2\sqrt{(1-p)px}+p(1-x)$ is nondecreasing for $0\le x\le1$ when $p\le1/2$.
\begin{lemma}[Pairwise bound and syndrome budget]\label[lemma]{lem:budget}
    For $0<p<1/2$ and a fixed nonzero flow of weight $w$,
    \begin{equation}\label{ls:eq:pairwise}
     \Pr\{\wt{E-f}\le\wt E\}\le B_\Theta(p)^w.
    \end{equation}
    Let $u=Dp>0$. Uniformly over simple $D$-regular graphs, outside an event of probability at most
    \begin{equation}\label{ls:eq:budgettail}
     \eta_n=e^{-u^2n/D}+e^{-u^4n/(4D)},
    \end{equation}
    every favorable connected flow support has $w$ edges and $w-t$ vertices with
    \begin{equation}\label{ls:eq:budget}
     w\le2un,\qquad t\le2u^2n.
    \end{equation}
\end{lemma}
\begin{proof}
    On an edge where $f_e\ne0$, the score increment is
    \begin{equation}
     S_e=\mathbb{1}_{\{E_e=f_e\}}-\mathbb{1}_{\{E_e=0\}}.
    \end{equation}
    Its probabilities at $1$ and $-1$ are $p\theta_e(f_e)$ and $1-p$. Set
    \begin{equation}\label{eq:ml-hamming-tilt}
     t_\Theta=\frac12\log\frac{1-p}{p\Theta}>0.
    \end{equation}
    For every edge in the support of $f$,
    \begin{equation}
     \mathbb E e^{t_\Theta S_e}
     =(1-p)e^{-t_\Theta}
       +p\theta_e(f_e)e^{t_\Theta}+p(1-\theta_e(f_e))
     \le B_\Theta(p).
    \end{equation}
    The same tilt works for all flow values, so independence and Markov's inequality prove \cref{ls:eq:pairwise}.
    
    For a deterministic error $e$, put $\mathcal D(e)=2\wt e-\wt{B^Te}$. Let $e'=e-f$, and write
    \begin{equation}
     R(x)\coloneqq 2\wt x-|V(\supp x)|=\sum_v(\deg_{\supp x}(v)-1)_+.
    \end{equation}
    A nonzero flow support has minimum degree at least two. Since its edges are contained in $\supp e\cup\supp e'$, a vertexwise comparison gives
    \begin{equation}
     2t=\sum_v(\deg_{\supp f}(v)-2)_+\le R(e)+R(e').
    \end{equation}
    Every nonzero syndrome coordinate is incident to both error supports, and favorability gives $\wt{e'}\le\wt e$. Hence $R(e),R(e')\le\mathcal D(e)$, so $t\le\mathcal D(e)$.
    
    A vertex incident to exactly one corrupted edge has nonzero syndrome. Consequently, its probability $s$ of a nonzero syndrome satisfies
    \begin{equation}
     s\ge Dp(1-p)^{D-1},\qquad
     Dp-s\le Dp\bigl[1-(1-p)^{D-1}\bigr]
           \le D(D-1)p^2\le u^2.
    \end{equation}
    Therefore $\E\mathcal D(E)\le u^2n$, for every replacement distribution. Changing one error symbol changes $\mathcal D$ by at most four. The bounded-differences inequality gives the second exponential in \cref{ls:eq:budgettail}; applied to $\wt E$, it gives $\Pr(\wt E>un)\le e^{-u^2n/D}$. Finally $\wt f\le\wt E+\wt{E-f}\le2\wt E$.
\end{proof}

\subsubsection{A connected-flow enumerator}
Let $A_{w,t}(G)$ count nonzero flows whose connected support has $w$ edges and $w-t$ vertices. The next bound keeps the excess $t$ explicit.

\begin{lemma}\label[lemma]{lem:enumerator}
For $w\le n/4$, set $a=w/n$ and
\begin{equation}
 L_w=w\sqrt{\frac{2qw}{n-w}}.
\end{equation}
Then
\begin{equation}\label{ls:eq:enumerator}
 \E[\mathbb1_{\mathcal A_n}A_{w,t}]
 \le\frac{qD}{D-1}\,(D-1)^w e^{(4a+\xi_n)w}
       \frac{L_w^{2t}}{(2t)!}.
\end{equation}
\end{lemma}
\begin{proof}
Put $s=w-t$. A connected support has flow-space dimension $w-s+1=t+1$, so it supports at most $q^{t+1}$ flows. For degrees $d_1,\ldots,d_s\ge2$ with sum $2w$, pairing half-edges bounds the number of simple realizations by $(2w-1)!!/\prod_i d_i!$. Coefficientwise,
\begin{equation}
 \sum_{d\ge2}\frac{z^d}{d!}\preceq\frac{z^2}{2}e^z,
\end{equation}
so summing over these degree sequences costs at most
\begin{equation}
 (2w-1)!!\,2^{-s}\frac{s^{2t}}{(2t)!}.
\end{equation}
Apply \cref{lem:subgraph}, including its connected port-labeling bound, to obtain
\begin{equation}\label{ls:eq:rawenumerator}
 \E[\mathbb1_{\mathcal A_n}A_{w,t}]
 \le q^{t+1}D(D-1)^{w-1}
 \frac{e^{\xi_n w}}{(n-2w)^w}
 \binom n{w-t}(2w-1)!!\,2^{-(w-t)}
 \frac{(w-t)^{2t}}{(2t)!}.
\end{equation}
We applied the coloring bound only to connected supports; subsequently allowing all pairings merely overcounts them. Now use
\begin{equation}
 \frac{\binom n{w-t}}{\binom nw}\le\left(\frac{w}{n-w}\right)^t,
 \qquad
 \frac{\binom nw(2w-1)!!2^{-w}}{n^w}\le1,
 \qquad -\log(1-2a)\le4a.
\end{equation}
These inequalities turn \cref{ls:eq:rawenumerator} into \cref{ls:eq:enumerator}.
\end{proof}

\subsubsection{The reliability criterion}
Write $R=(D-1)B_\Theta(p)$ and define
\begin{equation}\label{ls:eq:parameters}
 u=Dp,\quad K=2qu^2,\quad
 a_*=\frac{4K^{2/3}}q,\quad
 \Psi(a)=4a+\sqrt{\frac{2qa}{1-a}},\quad
 b_0=\frac{1+\log8}{2}.
\end{equation}

\begin{proposition}\label[proposition]{prop:criterion}
For fixed $p>0$, recovery is reliable if
\begin{equation}\label{ls:eq:criterion}
 u<\frac18,\quad a_*\le\frac14,\qquad
 \log R+\Psi(a_*)<0,\qquad
 \log R+8u+b_0K^{1/3}<0.
\end{equation}
\end{proposition}
\begin{proof}
On the graph event $\mathcal A_n$, apply \cref{ls:eq:pairwise} to the enumerator. For $w\le a_*n$, summing over all excesses costs at most $e^{L_w}$, giving
\begin{equation}
 \sum_t\E[\mathbb1_{\mathcal A_n}A_{w,t}]B_\Theta(p)^w
 \le\frac{qD}{D-1}
       \exp\{w[\log R+\Psi(a_*)+\xi_n]\}.
\end{equation}
All nonzero flows have weight at least $g\to\infty$, so the sum of these bounds tends to zero.

For the remaining weights allowed by \cref{ls:eq:budget}, put $a=w/n$ and $y=qa$. Since $a\le2u<1/4$, we have $L_w\le2w\sqrt y$. Set $T_0=\lfloor2u^2n\rfloor$ and $h=K/y^{3/2}$. For $a\ge a_*$, $h\le1/8$, and
\begin{align}
 \sum_{t=0}^{T_0}\frac{L_w^{2t}}{(2t)!}
 &\le h^{-2T_0}e^{hL_w},\\
 \frac1w\log(h^{-2T_0}e^{hL_w})
 &\le\frac{2K}{y}\left(1+\log\frac{y^{3/2}}K\right)
 \le b_0K^{1/3}.\label{ls:eq:truncation}
\end{align}
The last function is decreasing for $y\ge4K^{2/3}$, and its value at that endpoint is $b_0K^{1/3}$. Thus each remaining weight contributes at most
\begin{equation}
 \frac{qD}{D-1}\exp\{w[\log R+8u+b_0K^{1/3}+\xi_n]\}.
\end{equation}
Summing from $w\ge a_*n$ gives an exponentially small bound. If that range is empty, nothing is required there.

Finally add $\Pr(\mathcal A_n^c)=o(1)$ and the channel-budget probability \cref{ls:eq:budgettail}, and divide by $\Pr(\mathcal E_n)=1-o(1)$. No independence between the budget and favorability has been assumed: the budget was used only to restrict the indices of a union bound.
\end{proof}

\begin{lemma}[Recovery after one-coordinate changes]
\label{lem:robust-minweight-recovery}
Under \cref{ls:eq:criterion}, with probability tending to one over the Linial--Simkin graph and the errors, every word at Hamming distance at most one from $E$ is the unique minimum-weight representative of its syndrome. In particular, a fixed lexicographic minimum-weight syndrome decoder corrects all these words.
\end{lemma}
\begin{proof}
For a connected nonzero flow $f$, write $Q_f(E)=\wt E-\wt{E-f}$. Changing one error coordinate changes $Q_f$ by at most two. If a word at distance at most one from $E$ is not uniquely minimum-weight, a connected component of a competing flow therefore satisfies $Q_f(E)\ge-2$. The common tilt in \cref{eq:ml-hamming-tilt} gives
\begin{equation}\label{eq:ml-near-favorable}
 \Pr\{Q_f(E)\ge-2\}\le e^{2t_\Theta}B_\Theta(p)^w.
\end{equation}
For such a flow, $\wt{E-f}\le\wt E+2$. The deterministic comparisons in \cref{lem:budget} now give
\begin{equation}
 w\le2\wt E+2,\qquad
 t\le\mathcal D(E)+2.
\end{equation}
Outside the same exceptional event of probability $\eta_n$, the budget is therefore $w\le2un+2$, $t\le2u^2n+2$. Repeat the union bound in \cref{prop:criterion}. The factor $e^{2t_\Theta}$ is fixed, and the additive changes to the budget contribute $o(1)$ to the exponent per edge. The two strict exponent inequalities in \cref{ls:eq:criterion} still make the sum tend to zero. Since $u<1/8$, the enlarged weight range lies below $n/4$ for all sufficiently large $n$.
\end{proof}

\begin{proof}[Proof of the lower bound in \cref{thm:ml-small-alphabet}]
Return to a common even distribution $\theta$, so $\Theta=\theta_\mathrm{max}$. Let $\gamma=q/D^{2/3}\le1/32$, put $\delta=\delta_\mathrm{cyc}(D,\theta)$, and set $p=(1-16\gamma)\delta$. Since $q\ge2$, we have $D\ge512$. The lower bound in \cref{eq:ml-affinity-comparison} and $\theta_\mathrm{max}\ge1/(q-1)$ show that
\begin{equation}\label{eq:ml-general-small-cycle}
 \delta\le\frac q{D^2}<\frac12,\qquad
 (D-1)Z_\theta(\delta)=1.
\end{equation}
Indeed, at $p=q/D^2$ the lower bound $2\sqrt{p(1-p)/(q-1)}$ is larger than $1/(D-1)$. Consequently,
\begin{equation}\label{ls:eq:elementary}
 u\le\gamma D^{-1/3}\le\frac1{256},\qquad
 K^{1/3}<\frac43\gamma,\qquad
 a_*\le\frac{64\gamma^2}{9q}\le\frac1{288}.
\end{equation}
By \cref{eq:ml-affinity-comparison}, $(D-1)B_\Theta(\delta)\le1+(D-1)\delta$. Comparing the square-root and linear terms in \cref{eq:ml-hamming-affinity} gives
\begin{align}
 \frac{B_\Theta((1-16\gamma)\delta)}{B_\Theta(\delta)}
 &\le\sqrt{\frac{1-16\gamma}{1-\delta}},\\
 \log R
 &\le-8\gamma+\frac{\delta}{2(1-\delta)}+(D-1)\delta
 <-7\gamma.
\end{align}
Here $\delta\le\gamma D^{-4/3}$ and $(D-1)\delta\le\gamma D^{-1/3}\le\gamma/8$. The remaining terms in \cref{ls:eq:criterion} obey
\begin{equation}
 \frac{4a_*}{\gamma}\le\frac49,\qquad
 \frac1\gamma\sqrt{\frac{2qa_*}{1-a_*}}<3.8,\qquad
 \Psi(a_*)<5\gamma,
\end{equation}
and
\begin{equation}
 8u+b_0K^{1/3}
 \le\gamma+\frac43b_0\gamma<4\gamma.
\end{equation}
Thus \cref{ls:eq:criterion} holds with strict slack, giving recovery and the stronger conclusion of \cref{lem:robust-minweight-recovery}. All its left-hand sides are nondecreasing functions of $p$ in this interval, so the result also holds at every smaller positive $p$. At $p=0$, growing girth gives the conclusion directly.

Finally, at the cycle crossing, \cref{eq:ml-affinity-comparison} gives
\begin{equation}
 \frac{1}{D-1}
 =2\sqrt{\delta(1-\delta)\theta_\mathrm{max}}+O(\delta).
\end{equation}
Uniformly for $q=o(D^{2/3})$, \cref{eq:ml-general-small-cycle} implies $(D-1)\delta=o(1)$ and hence $\delta=(1+o(1))/(4\theta_\mathrm{max}D^2)$. Combining the lower bound just proved with the strong converse proves \cref{eq:ml-small-asymptotic}.
\end{proof}

\begin{lemma}[From channel reliability to a cutoff]
\label{lem:robust-channel-cutoff}
Fix independent nonzero-symbol distributions $\theta_1,\ldots,\theta_m$ and use the lexicographically first minimum-Hamming-weight representative of each syndrome. Let $R_k$ be the probability under $P_k$ that the decoder fails on at least one word at Hamming distance at most one from the sampled error. Then $R_k$ is nondecreasing in $k$. For independent errors with $\Pr(E_i=0)=1-p'$ and $\Pr(E_i=a)=p'\theta_i(a)$, write $R(p')$ for the same failure probability. If $K\sim\operatorname{Bin}(m,p')$, then
\begin{equation}\label{eq:robust-channel-cutoff-transfer}
 \max_{k\le\ell}R_k=R_\ell
 \le\frac{R(p')}{\Pr(K\ge\ell)}.
\end{equation}
In particular, for fixed $p<p'$, a vanishing $R(p')$ implies uniformly vanishing $R_k$ through $\ell=\lfloor pm\rfloor$.
\end{lemma}
\begin{proof}
The selected representatives are closed under deleting nonzero coordinates. Indeed, replacing a nonzero coordinate by zero increases the weight disadvantage of any competing representative, unless the competing word agrees on that coordinate. In the latter case the lexicographic comparison is unchanged. The stronger set of words whose one-coordinate neighbors are all selected is also closed under deletion: a change on a different coordinate commutes with deletion, and a change on the deleted coordinate is already a neighbor of the original word.

To couple $P_k$ and $P_{k+1}$, first draw independent nonzero symbols from the $\theta_i$, choose a uniform ordering of the coordinates, and retain the first $k$ or $k+1$ symbols. Deletion closure gives $R_k\le R_{k+1}$. Conditioning the independent channel on $K$ gives $R(p')=\mathbb E R_K\ge R_\ell\Pr(K\ge\ell)$.
\end{proof}


\subsection{Planting a small matching early in a round}
\label{ls:sec:planting}
The converse needs a lower probability bound. Throughout this section, $H$ is an arbitrary $(D-1)$-regular graph of girth at least $g$. Fix an observation time $T$ in the final round such that
\begin{equation}\label{ls:eq:earlyclock}
    n_* \coloneqq n - 2T = n^{(1+c)/2+o(1)}
\end{equation}
and let $M_T$ be the matching selected in the first $T$ steps. A graph of maximum degree $D$ has at most $C_D(D-1)^g = O_D(n^c)$ vertices in a ball of radius $g-2$. Since $n_* \gg n^c$, the process cannot freeze before $T$. Thus $n_*$ is exactly the number of vertices left unmatched at time $T$.
\begin{lemma}[Early matching probabilities]\label[lemma]{lem:planting}
Uniformly over $H$ and specified matchings $F$ with $s=|F|\le K_0\log n$,
\begin{equation}\label{ls:eq:earlyupper}
 \Pr(F\subseteq M_T\mid H)\le(1+o(1))n^{-s}.
\end{equation}
If in addition $H+F$ is simple and has girth at least $g$, then
\begin{equation}\label{ls:eq:planting}
 \Pr(F\subseteq M_T\mid H)=(1+o(1))n^{-s}.
\end{equation}
Here $K_0$ is any fixed constant.
\end{lemma}
\begin{proof}
    The upper bound follows from the potential in \cref{lem:matching}, using the state-dependent missing-degree bound $C_Dn^c$. Stop at $T$, then complete the residual matching uniformly without a girth constraint for this calculation only. The cumulative multiplicative error is
    \begin{equation}\label{ls:eq:earlydrift}
     \exp\left(O_D\left(s n^c\sum_{v\ge n_*}v^{-2}\right)\right) = \exp(O_D(sn^c/n_*))=1+o(1).
    \end{equation}
    Together with $\Phi_s(n)=(1+o(1))n^{-s}$ this proves \cref{ls:eq:earlyupper}.
    
    For the lower bound, we use a planted process. Choose a uniform injection from the $s$ labeled edges of $F$ to the first $T$ selection times. Insert each prescribed edge at its assigned time. At every other step select a uniform available pair avoiding all endpoints of still-pending edges of $F$. For purposes of defining a probability measure, insert a forced edge even when it is unavailable; call such a trajectory \emph{invalid}. Forced endpoints are either protected or already matched, so no unforced edge ever touches $V(F)$. Even on invalid trajectories the graph is simple and has maximum degree $D$, and the unforced choices exist because at least $n_*-2s\gg n^c$ outside vertices remain free.

    \emph{Likelihood ratio.} A valid original trajectory containing $F$ determines its forcing schedule uniquely. At a state with $v$ free vertices and $k$ pending prescribed edges, let
    \begin{equation}
     N=\binom v2,\quad N^{\circ}=\binom{v-2k}{2},\qquad
     A=|\mathscr A_t|,\quad A^{\circ}=\#\{\text{available outside pairs}\}.
    \end{equation}
    The original-to-planted likelihood ratio on that trajectory is
    \begin{equation}\label{ls:eq:likelihood-product}
     (T)_s\prod_{\text{forced times}}\frac1A
           \prod_{\text{unforced times}}\frac{A^{\circ}}A.
    \end{equation}
    In the complete graph, this product telescopes to
    \begin{equation}\label{ls:eq:reference}
     p_{\mathrm{ref}}=\frac{(T)_s2^s}{(n)_{2s}}=(1+o(1))n^{-s},
    \end{equation}
    independently of the schedule. For example, the identity follows by writing each outside ratio as $(v-2k)(v-2k-1)/(v(v-1))$; its numerator cancels successive outside-vertex counts.
    
    Put $\Delta=C_Dn^c$. At a forced step, replacing $N$ by $A$ changes the logarithm of the product by $O_D(\Delta/v)$. At an unforced step write $A=N-M$ and $A^{\circ}=N^{\circ}-M^{\circ}$. Since
    \begin{equation}
        M-M^{\circ}\le2k\Delta,\qquad M^{\circ}\le(v-2k)\Delta/2,
    \end{equation}
    we have
    \begin{equation}
        \left|\log\frac{A^{\circ}/A}{N^{\circ}/N}\right| = O_D(k\Delta/v^2).
    \end{equation}
    Summing these estimates shows uniformly over valid trajectories that
    \begin{equation}\label{ls:eq:ratio}
        \frac{\Pr_{\mathrm{original}}(\text{trajectory})}{\Pr_{\mathrm{planted}}(\text{trajectory})} = p_{\mathrm{ref}}\exp(O_D(sn^c/n_*))=(1+o(1))n^{-s}.
    \end{equation}
    
    \emph{Validity of the planted process.} Let $U$ be its unforced matching. For any prescribed matching $U_0$ of $t=O(\log n)$ edges outside $V(F)$,
    \begin{equation}\label{ls:eq:unforced}
        \Pr_{\mathrm{planted}}(U_0\subseteq U)\le(1+o(1))n^{-t}.
    \end{equation}
    To see this, fix the forcing schedule and apply the matching potential on the $n-2s$ outside vertices. Forced steps leave that potential unchanged; every unforced step has missing degree at most $\Delta$, and at least $n_*-2s$ outside vertices remain. The accumulated error is $\exp(O_D(t\Delta/(n_*-2s)))=1+o(1)$. Uniform artificial completion of the outside matching gives \cref{ls:eq:unforced}; averaging over schedules preserves it. No assertion about the girth of a possibly invalid trajectory is needed for this degree bound.

    Set $J=H+F$. If a forced edge $xy$ is unavailable when inserted, then $(J-xy)+U$ contains a simple $x$--$y$ path of length at most $g-2$. Such a path must use at least one edge of $U$, since $g_J\ge g$. We bound all these paths by their \emph{last} unforced edge.
    
    For precision, give each unforced edge in a path weight $1/n$, and each edge of $J-xy$ weight one. The total weight of nonbacktracking prefixes of length $a$ is at most $C_D(D-1)^a$. At a vertex outside $V(F)$, arrival along an old edge leaves at most $D-2$ old continuations and one unforced port, whose destinations have total weight at most one; arrival along an unforced edge leaves $D-1$ old continuations and no second unforced port. At a vertex of $V(F)$ there are at most $D-1$ old continuations and no unforced port. These observations prove the prefix bound by induction, even if global self-avoidance constraints are dropped.
    
    A suffix of length $b$ after the last unforced edge lies in $J-xy$ and has at most $C_D(D-1)^b$ choices when explored backwards from $y$. The intervening unforced edge is fixed by the two endpoints and contributes $1/n$. Summing over its position, and using \cref{ls:eq:unforced} for each simple path, bounds the probability of an alternate path of length $\ell$ by
    \begin{equation}
        O_D(\ell (D-1)^{\ell}/n).
    \end{equation}
    The prescribed unforced edges of a simple path form a matching; paths that violate this contribute zero. Summing over $\ell\le g-2$ and over the $s$ forced edges gives
    \begin{equation}\label{ls:eq:invalid}
        \Pr_{\mathrm{planted}}(\text{invalid}) \le O_D(sg (D-1)^g/n)=n^{c-1+o(1)}=o(1).
    \end{equation}
    Finally sum \cref{ls:eq:ratio} over valid trajectories. \Cref{ls:eq:reference,ls:eq:invalid} prove the lower bound.
\end{proof}

We also record an upper bound for the partial graph used in the converse. Let $\mathcal B_n$ denote the density event for the first $D-3$ matching rounds (the sure event when $D=3$), and let
\begin{equation}
    G_n^T \coloneqq H\cup M_1\cup\cdots\cup M_{D-3}\cup M_T.
\end{equation}
Combining the proof of \cref{lem:subgraph} for the preceding rounds with \cref{ls:eq:earlyupper} for the final prefix gives, for every specified port-labeled $J$ with $w\le K_0\log n$,
\begin{equation}\label{ls:eq:prefix-subgraph}
    \Pr(J\subseteq G_n^T,\mathcal B_n)\le(1+o(1))n^{-w}.
\end{equation}
The error is uniform in $J$, including disconnected $J$. Only upper bounds are being multiplied here, conditional on the graph already constructed.

\subsection{A supply of logarithmic-length cycles}\label{ls:sec:cycles}
Choose once and for all
\begin{equation}\label{ls:eq:L}
    c<\kappa<1,\qquad L = 2\lfloor\kappa\log_{D-1} n/2\rfloor,\qquad s_0 = \left\lceil\frac{3\log L}{\log((D-1)/(D-2))}\right\rceil.
\end{equation}
Then $s_0=O_D(\log\log n)$, $2s_0<g$, and $L\ge g$ for large $n$.

Condition on the $(D-1)$-regular graph $A$ present before the final round. A \emph{candidate cycle} consists of $k$ oriented paths in $A$, of lengths
\begin{equation}\label{ls:eq:runs}
    1\le b_i\le s_0,\qquad \sum_{i=1}^k(b_i+1)=L,
\end{equation}
whose vertex sets are pairwise at distance at least $g$, joined cyclically by $k$ proposed final-round matching edges. Two descriptions differing only by a cyclic starting path or reversal describe the same candidate.

The joining edges form a matching $F$, and $A+F$ has girth at least $g$. Indeed, an old path between distinct path groups has length at least $g$. A shorter new cycle would therefore have to travel within each group between its two endpoints and traverse the whole ring of joining edges. Each chosen path is the unique shortest connection between its endpoints, because its length is at most $s_0<g/2$. Such a ring has length at least $L\ge g$.

\begin{lemma}[Candidate-cycle first moment]\label[lemma]{lem:cyclemean}
    Uniformly over $A$, the expected number of candidate cycles whose joining edges lie in $M_T$ is
    \begin{equation}\label{ls:eq:cyclemean}
     (1+o(1))\frac{(D-1)^L}{2L}.
    \end{equation}
\end{lemma}
\begin{proof}
    For $b\le s_0$, the number of oriented length-$b$ paths in $A$ is exactly $n(D-1)(D-2)^{b-1}$. There is no self-intersection at these lengths. When choosing the paths in \cref{ls:eq:runs}, the separation condition discards a fraction at most
    \begin{equation}
        O_D(L^2(D-1)^{g+s_0}/n)=n^{c-1+o(1)}=o(1).
    \end{equation}
    For example, if a new path comes within distance $g-1$ of an earlier one, its starting vertex is within distance $g+s_0$ of an already selected vertex. A ball count gives the displayed bound.
    
    For a fixed ordered tuple of paths, \cref{lem:planting} gives joining probability $(1+o(1))n^{-k}$. This cancels the $n^k$ in the path count. Each candidate has exactly $2k$ ordered, oriented descriptions.
    
    Without the upper cutoff on run lengths, introduce
    \begin{equation}
     B(x)=\sum_{b\ge1}(D-1)(D-2)^{b-1}x^{b+1}
         =\frac{(D-1)x^2}{1-(D-2)x}.
    \end{equation}
    The resulting weighted count is
    \begin{align}
        [x^L]\frac12\sum_{k\ge1}\frac{B(x)^k}{k}
        &= [x^L]\left[-\frac12\log(1-B(x))\right]\notag\\
        &=\frac{(D-1)^L+(-1)^L-(D-2)^L}{2L} = (1+o(1))\frac{(D-1)^L}{2L}.\label{ls:eq:cycle-series}
    \end{align}
    Here $1-(D-2)x-(D-1)x^2=(1-(D-1)x)(1+x)$.
    
    To justify truncation, let
    \begin{equation}
        B_{>s_0}(x)=\frac{(D-1)(D-2)^{s_0}x^{s_0+2}}{1-(D-2)x}.
    \end{equation}
    Marking one long run bounds the discarded coefficient by
    \begin{equation}
        [x^L]\frac{B_{>s_0}(x)}{2(1-B(x))} \le C_D(D-1)^L\left(\frac{D-2}{D-1}\right)^{s_0}.
    \end{equation}
    Its ratio to \cref{ls:eq:cycle-series} is $O_D(L[(D-2)/(D-1)]^{s_0})=o(1)$. All errors above are uniform over the path lengths and over $A$, proving \cref{ls:eq:cyclemean}.
\end{proof}

\subsection{The strong converse}\label{ls:sec:converse}
Write $\nu=\nu_p$ for the noise distribution in \cref{eq:ml-general-noise}. Let the independent edge errors have this common distribution, with $\nu(z)=\nu(-z)$, and choose $a\ne0$ maximizing
\begin{equation}\label{ls:eq:converse-Z}
    Z=\sum_{z\in\mathbb F_q}\sqrt{\nu(z)\nu(z-a)},\qquad R=(D-1)Z>1.
\end{equation}
For each simple cycle, fix an orientation using its labeled edges and assign the circulation $\mathbf f$ with coefficients $\pm a$. We call the cycle favorable when its likelihood ratio satisfies
\begin{equation}\label{ls:eq:cycle-likelihood-ratio}
 L_C(\mathbf E)=\prod_{e\in C}
 \frac{\nu(E_e-f_e)}{\nu(E_e)}\ge1.
\end{equation}
The denominators are positive almost surely; a zero numerator sets the ratio to zero. Evenness of $\nu$ makes the score distributions for coefficients $a$ and $-a$ identical. Thus the probability $p_L$ of favorability depends only on the cycle length $L$.

\begin{lemma}[Single-cycle and joint score bounds]\label[lemma]{lem:scores}
For some fixed $C_0>0$ and all sufficiently large even $L$,
\begin{equation}\label{ls:eq:scorelower}
 Z^L L^{-C_0}\le p_L\le Z^L.
\end{equation}
For two distinct length-$L$ cycles sharing $j$ edges,
\begin{equation}\label{ls:eq:jointscore}
    \Pr\{L_{C_1}(\mathbf E)\ge1,\ L_{C_2}(\mathbf E)\ge1\} \le Z^{2(L-j)}.
\end{equation}
\end{lemma}
\begin{proof}
For $b\ne0$, let $S_b(z)=\log[\nu(z-b)/\nu(z)]$ on the support of $\nu$, allowing the value $-\infty$. Then
\begin{equation}\label{ls:eq:moments}
 \E e^{S_a(E_e)/2}=Z, \qquad \E e^{S_b(E_e)} = \sum_{z:\nu(z)>0}\nu(z-b)\le1.
\end{equation}
The first identity and Markov's inequality give the upper bound in \cref{ls:eq:scorelower}. For the lower bound, use the tilted distribution
\begin{equation}\label{ls:eq:tilted-law}
 \pi_a(z)=Z^{-1}\sqrt{\nu(z)\nu(z-a)}.
\end{equation}
The involution $z\mapsto a-z$ preserves $\pi_a$ and reverses $S_a$. Choose symbol counts within $O(q)$ of $L\pi_a(z)$, equal on paired orbits. Such counts summing to $L$ exist: round the counts on each pair down equally, round fixed-point counts down to even integers, and allocate the even remainder to an orbit with positive $\pi_a$-mass. The resulting type has total score zero. Stirling's formula bounds its tilted probability below by a reciprocal polynomial in $L$. Undoing the tilt multiplies this probability by $Z^L$, proving \cref{ls:eq:scorelower}, including when $\nu$ has zero entries.

On a shared edge, Cauchy--Schwarz and \cref{ls:eq:moments} give
\begin{equation}
 \E e^{[S_b(E_e)+S_{b'}(E_e)]/2}
 \le\sqrt{\E e^{S_b(E_e)}\E e^{S_{b'}(E_e)}}\le1.
\end{equation}
Each unshared edge contributes $Z$. Markov's inequality applied to $\sqrt{L_{C_1}L_{C_2}}$ proves \cref{ls:eq:jointscore}.
\end{proof}

Let $X$ count favorable candidate cycles already present in $G_n^T$, and put
\begin{equation}
 Y=\mathbb1_{\mathcal B_n}X,\qquad
 \mu_L=\frac{(D-1)^L}{2L}p_L.
\end{equation}
The channel can be sampled independently on all potential edges before generating the graph. Since the candidate length in \cref{ls:eq:L} is even, \cref{lem:cyclemean,lem:scores} give
\begin{equation}\label{ls:eq:firstmoment}
 \E Y=(1+o(1))\mu_L,\qquad
 \mu_L\ge\frac{R^L}{2L^{C_0+1}}\rightarrow\infty.
\end{equation}
The cycle count is uniform over the graph before the last round, so the preceding-round event $\mathcal B_n$ costs only its probability $1-o(1)$.

\subsubsection{Disjoint cycles}
Let $P$ be the permutation matrix of the port involution, and let $J_D$ be the all-ones matrix. The transition matrix for a nonbacktracking port word is $Q=J_D-P$. Its eigenvalues are $D-1$, $1$, and $-1$ with multiplicity $D-2$. Hence
\begin{equation}\label{ls:eq:trace}
    \tr(Q^L)=(D-1)^L+1+(D-2)(-1)^L=(1+o(1))(D-1)^L.
\end{equation}
For two vertex-disjoint cycles the channel events are independent. There are $(n)_{2L}/(2L)^2$ ordered pairs of labeled vertex cycles. Applying \cref{ls:eq:prefix-subgraph} to their union and \cref{ls:eq:trace} to each cycle bounds their contribution to $\E Y^2$ by
\begin{equation}\label{ls:eq:disjoint}
    (1+o(1))\frac{(n)_{2L}}{(2L)^2}n^{-2L}(D-1)^{2L}p_L^2 = (1+o(1))\mu_L^2.
\end{equation}

\subsubsection{Overlapping cycles}
Consider distinct length-$L$ cycles with $j$ common edges. Their intersection consists of paths and isolated vertices. If it has $k\ge1$ components, their union has
\begin{equation}\label{ls:eq:unionparameters}
    w=2L-j\quad\text{edges},\qquad v=2L-j-k\quad\text{vertices}.
\end{equation}
There are at most $(CL)^{Ck}$ abstract intersection patterns: record the endpoints of the common paths on both cycles, pair the components, and record their orientations. Isolated vertices use coincident endpoints. This bound is independent of $j$.

For each pattern there are at most $n^v$ embeddings and $D(D-1)^{w-1}$ connected port-labelings. By \cref{ls:eq:prefix-subgraph}, its graph contribution, including $\mathcal B_n$, is at most
\begin{equation}
    (1+o(1))\frac{D}{D-1}\,n^{-k}(D-1)^{2L-j}.
\end{equation}
Let $J$ count unordered favorable pairs sharing a vertex, setting $J=0$ outside $\mathcal B_n$. Combining this bound with \cref{ls:eq:jointscore,ls:eq:firstmoment} and summing over $j$ gives
\begin{equation}\label{ls:eq:overlap}
 \frac{\E J}{\mu_L^2}
 \le L^{C_1}\sum_{k=1}^{L}
       \left(\frac{CL^{C_1}}n\right)^k
       \max_{0\le j<L}\left((D-1)Z^2\right)^{-j}
 =o(1).
\end{equation}
Indeed, $(D-1)Z>1$ implies $((D-1)Z^2)^{-1}<D-1$, so \cref{ls:eq:L} bounds the maximum by $(D-1)^L\le n^\kappa$. The sum is therefore $n^{\kappa-1+o(1)}=o(1)$.

\subsubsection{Disjoint favorable cycles and decoding failure}
The diagonal contribution to $\E Y^2$ is $\E Y=o(\mu_L^2)$. Together with \cref{ls:eq:firstmoment,ls:eq:disjoint,ls:eq:overlap}, this gives
\begin{equation}\label{ls:eq:cycle-concentration}
 \frac{Y}{\mu_L}\rightarrow1,\qquad
 \frac{J}{\mu_L^2}\rightarrow0
 \quad\text{in probability}.
\end{equation}
Join two favorable cycles when they share a vertex. An independent set in this intersection graph has size at least $Y^2/(Y+2J)$ when $Y>0$. This bound follows by ordering its vertices uniformly at random and retaining each vertex that precedes all its neighbors: the expected number retained is $\sum_i(d_i+1)^{-1}\ge Y^2/(Y+2J)$. Thus, if $K_n$ is the maximum number of vertex-disjoint favorable candidate cycles, set to zero outside $\mathcal B_n$, then
\begin{equation}\label{ls:eq:disjoint-favorable-count}
 K_n\rightarrow\infty\quad\text{in probability}.
\end{equation}

On a completed graph, subtracting any subset of $k$ disjoint favorable circulations from the sampled error gives $2^k$ distinct errors with the same syndrome, each with probability at least that of the sampled error. Its posterior probability within the syndrome is therefore at most $2^{-k}$. Splitting according to whether a decoder's selected error has $K_n\ge k$ gives, for every fixed positive integer $k$,
\begin{equation}\label{ls:eq:posterior-converse}
 \E[\PS(G_n,p,\theta)\mid\mathcal E_n]
 \le\frac{\Pr(K_n<k)}{\Pr(\mathcal E_n)}+2^{-k}.
\end{equation}
The favorable circulations in $G_n^T$ remain circulations after later edges are added. Taking $n\to\infty$ in \cref{ls:eq:posterior-converse} and then $k\to\infty$ proves \cref{thm:ml-general-converse}, including arbitrary likelihood ties.

For the $q$SC, $Z=Z_q(p)$ and $(D-1)Z>1$ throughout $\dc(D,q)<p\le (q-1)/q$. Hence the converse holds on the full $q$SC range. In particular, at $p=(q-1)/q$ the noise is uniform and $\PS(G_n,(q-1)/q)=q^{-(m-n+1)}$.

\subsection{The binary endpoint and scope of the conclusion}\label{ls:sec:binary}
For $q=2$, a flow support is an even subgraph and decomposes into edge-disjoint simple cycles. A favorable flow therefore contains a favorable cycle. If $R=(D-1)Z_2(p)<1$, \cref{ls:eq:smallsubgraph,ls:eq:trace} give, uniformly for $g\le\ell\le K_0\log n$,
\begin{equation}
 \E[\mathbb1_{\mathcal A_n}C_\ell(G_n)]
 \le(1+o(1))\frac{(D-1)^\ell+O_D(1)}{2\ell}.
\end{equation}
The resulting cycle union bound is $O_D(R^g/g)=o(1)$. For longer cycles use
\begin{equation}
 C_\ell(G)\le\frac{nD(D-1)^{\ell-1}}{2\ell}.
\end{equation}
Choosing $K_0$ sufficiently large makes the tail $O_D(nR^{K_0\log n})=o(1)$. Adding the graph bad-event probability proves binary recovery below $\dc(D,2)$ at every degree. The converse gives \cref{eq:ml-binary-exact}.

This recovery guarantee also holds simultaneously after every single-coordinate modification of the error. A failure after such a modification implies a cycle whose Hamming comparison score on the original error is at least $-2$; at $t=\frac12\log[(1-p)/p]$, its probability is at most $e^{2t}Z_2(p)^\ell$, and the same cycle sums tend to zero.

\subsection{The remaining fixed-parameter question}
For uniform replacement, the cycle--entropy conjecture below is not settled by these asymptotic bounds. At fixed $G$, put
\begin{equation}\label{eq:ml-conditional-entropy}
 \Phi_n(G,p)\coloneqq\frac1mH(\mathbf E\mid B^T\mathbf E),\qquad
 \overline\Phi_n(p)\coloneqq\mathbb E_G\Phi_n(G,p).
\end{equation}
\begin{conjecture}[Cycle--entropy transition]
\label{conj:ml-cycle-entropy}
For every fixed $D\ge3$, finite field $\mathbb F_q$, and $c\in(0,1)$, $\overline\Phi_n(p)$ for \cref{def:linial-simkin-ensemble} converges for $0<p<(q-1)/q$ to a function $\Phi_{D,q,c}(p)$. Define
\begin{equation}\label{eq:ml-entropy-onset}
 \delta_\mathrm{ent}(D,q)
 =\inf\{p\in(0,(q-1)/q):\Phi_{D,q,c}(p)>0\}.
\end{equation}
Then, with
\begin{equation}\label{eq:ml-exact-conjecture}
 \delta_*=\delta_\mathrm{rel}(D,q)
 =\min\{\delta_\mathrm{cyc}(D,q),\delta_\mathrm{ent}(D,q)\},
\end{equation}
\begin{equation}\label{eq:ml-conjectured-step}
 \mathbb E_G P_{\mathrm S}^{\mathrm{ML}}(G,p)\rightarrow
 \begin{cases}
 1,&0<p<\delta_*,\\
 0,&\delta_*<p<(q-1)/q.
 \end{cases}
\end{equation}
No assertion is made at $p=\delta_*$.
\end{conjecture}

\section{Proof of the LP-decoding guarantee}
\label[appendix]{app:lp-reliability}

This appendix proves \cref{thm:lp-reliability} using the direct LP formulation in \cref{eq:lp-local-polytope}. The final subsection gives the more compact formulation announced in \cref{sec:lp-decoder}.

Let $G=(V,E)$ be the underlying simple graph, with $m=|E|$. Write $H=B^T$ and take $h_{ve}$ to denote the entry of the matrix $H$ indexed by vertex $v$ and edge $e$. Each edge represents a symbol in the codeword, and each vertex represents a parity check. If $e=\{u,v\}$, then column $e$ has nonzero entries in rows $u$ and $v$ and zeros in all other rows. For each vertex $v$, write
\begin{equation}
    \partial v=\{e\in E:h_{ve}\ne0\},\qquad d_v=|\partial v|.
\end{equation}
Thus $\partial v$ is the support of the parity check at $v$, or equivalently the set of edges that have $v$ as an endpoint. A symbol string $\mathbf b\in\mathbb F_q^E$ has syndrome $\mathbf s$ precisely when
\begin{equation}
    \sum_{e\in\partial v}h_{ve}b_e=s_v \qquad \text{for every }v\in V.
\end{equation}

Although \cref{thm:lp-reliability} is stated for homological cycle codes, the proof below applies to arbitrary cycle codes $\ker H$ over $\mathbb F_q$ with the column supports just specified. We do not require $h_{ue}+h_{ve}=0$ for $e=\{u,v\}$. Consequently, the certified rate in \cref{cor:lp-linial} also applies to nonhomological cycle codes whose underlying graphs belong to the Linial--Simkin ensemble.

For fixed maximum degree $D$, the direct LP already has polynomial size in $n=|V|$, $m$, and $q$: each nonempty check involves at most $D$ symbols and has at most $q^{D-1}$ satisfying assignments. The recovery proof below uses only this formulation, with the symbol marginals $x_{e,a}$, the local assignment probabilities $\lambda_{v,\mathbf b}$, and the objective $\mathcal L$ defined in \cref{sec:lp-decoder}.

We first compare a competing feasible LP solution with the true error vector. For each edge, we consider the total probability that the solution's symbol marginal assigns to values other than the true error value. We call these quantities \emph{disagreement probabilities}. The parity checks impose inequalities on these probabilities, which we use to construct a distribution on nonbacktracking walks in $G$. A competing solution with no larger objective value would force some such walk to have at least half of its positions corrupted. Walks shorter than the girth are simple paths, so counting paths and bounding their error probabilities proves the theorem. Only after completing this recovery argument do we introduce a reformulation that removes the exponential dependence of the LP size on $D$.

\subsection{Disagreement probabilities and the LP objective}
Fix an error vector $\mathbf e$ and its syndrome $\mathbf s=H\mathbf e$. Let $\mathbf x^{\mathbf e}$ be the corresponding integral feasible marginal vector, with entries $x^{\mathbf e}_{f,a}=\mathbb1_{\{a=e_f\}}$. For any other feasible marginal vector $\mathbf x$, define the disagreement probabilities by
\begin{equation}\label{eq:lp-disagreement-probability}
 z_f\coloneqq 1-x_{f,e_f}
 =\sum_{a\in\mathbb F_q\setminus\{e_f\}}x_{f,a},
\end{equation}
and put
\begin{equation}
 T\coloneqq \{f:e_f\ne0\},\qquad
 \omega_f\coloneqq 1-2\mathbb1_{\{f\in T\}}.
\end{equation}
Thus $T$ is the set of erroneous edges, and $\omega_f$ is $+1$ on an uncorrupted edge and $-1$ on a corrupted edge. The vector $\mathbf z$ is nonnegative, and $\mathbf z\ne0$ whenever $\mathbf x\ne\mathbf x^{\mathbf e}$.

Fix a vertex $v$ and sample a local assignment $\mathbf b$ from a distribution $\lambda_{v,\mathbf b}$ witnessing feasibility of $\mathbf x$. Both $\mathbf b$ and the restriction of $\mathbf e$ to $\partial v$ satisfy the check at $v$, so
\begin{equation}
    \sum_{f\in\partial v}h_{vf}(b_f-e_f)=0.
\end{equation}
These two assignments cannot differ on exactly one edge $f$ in the check support: that would give $h_{vf}(b_f-e_f)=0$ with both factors nonzero in $\mathbb F_q$. Therefore, if $b_f\ne e_f$, there must be another edge $f'\in\partial v\setminus\{f\}$ on which $b_{f'}\ne e_{f'}$. By local consistency, the probabilities of these disagreements are $z_f$ and $z_{f'}$. A union bound gives
\begin{equation}\label{eq:lp-mass-triangle}
 z_f\le\sum_{f'\in\partial v\setminus\{f\}}z_{f'}
\end{equation}
for every $v\in V$ and $f\in\partial v$. This is the only algebraic property of the checks used in the recovery argument. It uses nonzero entries within each check support, but no relation between the entries at the two endpoints of an edge. It also uses only one local assignment distribution at a time; no joint distribution on complete syndrome-consistent words is assumed.

On an uncorrupted edge, the objective difference is $1-x_{f,0}=z_f$. On a corrupted edge, it is $-x_{f,0}\ge-z_f$, because $e_f\ne0$ and $x_{f,0}\le1-x_{f,e_f}$. Summing these contributions gives
\begin{equation}\label{eq:lp-objective-lower}
 \mathcal L(\mathbf x)-|\mathbf e|
 \ge\sum_f\omega_f z_f.
\end{equation}
In particular, $\mathbf x\ne\mathbf x^{\mathbf e}$ and $\mathcal L(\mathbf x)\le|\mathbf e|$ imply that $\mathbf z$ is nonzero, satisfies \cref{eq:lp-mass-triangle}, and obeys $\sum_f\omega_fz_f\le0$. Equivalently,
\begin{equation}\label{eq:lp-disagreement-balance}
 \sum_{f:e_f\ne0}z_f\ge\sum_{f:e_f=0}z_f.
\end{equation}
The comparison applies to all feasible marginal vectors, including fractional solutions and equal-cost ties.

\subsection{From disagreement probabilities to walks}
A nonbacktracking walk traverses adjacent edges without immediately reversing the previous edge. We now use \cref{eq:lp-mass-triangle} to construct a random walk that has probability proportional to $z_f$ of traversing edge $f$ at each step. This converts the objective comparison into a statement about the error locations along walks.

\begin{lemma}[Walks with prescribed edge probabilities]\label{lem:lp-path-representation}
Let $\mathbf z\ne0$ be nonnegative edge weights satisfying \cref{eq:lp-mass-triangle}, and write $Z=\sum_fz_f$. For every integer $h\ge1$ there is a probability distribution on nonbacktracking walks $W=(F_1,\ldots,F_h)$ of $h$ edges with $\Pr(F_i=f)=z_f/Z$ for every step $i$. Consequently,
\begin{equation}\label{eq:lp-stationary-identity}
 \mathbb E_W\!\left[\sum_{i=1}^h\omega_{F_i}\right]
   =\frac hZ\sum_f\omega_fz_f
\end{equation}
for every real edge-weight vector $\boldsymbol\omega$. Repeated edges are counted with multiplicity.
\end{lemma}
\begin{proof}
For each vertex $v$, put $S_v=\sum_{f\in\partial v}z_f$. The inequalities \cref{eq:lp-mass-triangle} say that no edge in the check support has weight greater than $S_v/2$. We construct nonnegative numbers $b^v_{ff'}$, indexed by $f,f'\in\partial v$, that are symmetric in $f,f'$, vanish when $f=f'$, and sum to $z_f$ over $f'$. After arrival along an edge $f$ with $z_f>0$, the ratio $b^v_{ff'}/z_f$ will be the probability of continuing along $f'$.

If $S_v>0$, partition a circle of circumference $S_v$ into consecutive half-open intervals $I_f$ of lengths $z_f$, and let
\begin{equation}
 b^v_{ff'}=\left|I_f\cap(I_{f'}+S_v/2)\right|.
\end{equation}
Here addition is modulo $S_v$, and the absolute-value signs denote arc length. Rotating by half the circumference preserves length and is its own inverse, so $b^v_{ff'}=b^v_{f'f}$. An interval of length at most $S_v/2$ is disjoint from its half-turn except possibly at endpoints, so $b^v_{ff}=0$. The rotated intervals partition the circle, giving row sums $\sum_{f'\in\partial v}b^v_{ff'}=|I_f|=z_f$. Set $b^v=0$ when $S_v=0$.

Choose the first directed edge from the following distribution on both directions of each edge with $z_f>0$:
\begin{equation}
 \pi(u,v)=\frac{z_{\{u,v\}}}{2Z}.
\end{equation}
After traversing $f=\{u,v\}$ toward $v$, choose the next edge $f'=\{v,w\}$ with probability $b^v_{ff'}/z_f$. The row sums make this a Markov chain, and the zero diagonal prohibits immediate reversal. For each directed edge $(v,w)$, its incoming probability under $\pi$ is
\begin{equation}
 \sum_{u:\{u,v\}\in E}
 \frac{z_{\{u,v\}}}{2Z}
 \frac{b^v_{\{u,v\},\{v,w\}}}{z_{\{u,v\}}}
 =\frac{z_{\{v,w\}}}{2Z}=\pi(v,w),
\end{equation}
where terms with zero weight are omitted. Thus $\pi$ is \emph{stationary}: the distribution on directed edges is unchanged after a step. At every step, the probability of traversing the undirected edge $f$ is therefore $z_f/Z$. Summing the expected edge weights over the $h$ steps proves \cref{eq:lp-stationary-identity}.
\end{proof}

\begin{lemma}[A sufficient condition for unique LP recovery]\label{lem:lp-path-certificate}
Fix $h\ge1$. If every nonbacktracking walk of $h$ edges has strictly fewer than $h/2$ erroneous positions, counting multiplicity, then $\mathbf x^{\mathbf e}$ is the unique optimal marginal vector of the LP. In particular the LP decoder returns $\mathbf e$.
\end{lemma}
\begin{proof}
A walk's total weight is $h$ minus twice its number of erroneous positions, so the hypothesis makes every such weight strictly positive. If a distinct feasible marginal vector $\mathbf x$ satisfied $\mathcal L(\mathbf x)\le|\mathbf e|$, then \cref{eq:lp-objective-lower,lem:lp-path-representation} would give a distribution on these walks with nonpositive expected weight, a contradiction.
\end{proof}

This condition is sufficient, not necessary: the existence of a walk with at least half its positions corrupted does not by itself imply a decoding failure. The argument rules out both integral and fractional competitors, including equal-cost ties, and depends only on the error locations.

\subsection{Probability of a bad path}
\begin{proof}[Proof of \cref{thm:lp-reliability}]
    If $h<g_G$, every nonbacktracking walk of $h$ edges is a simple path: a first repeated vertex would yield a cycle of length at most $h$. There are at most $2m(D-1)^{h-1}$ oriented such paths, because the first edge has $2m$ orientations and each subsequent step has at most $D-1$ choices.
    
    For a fixed path, let $X$ count its erroneous positions. Under independent symbol errors, $X$ has the $\operatorname{Bin}(h,p)$ distribution. For $p<1/2$, Markov's inequality with $s=(1-p)/p>1$ yields
    \begin{equation}\label{eq:lp-path-tail}
     \Pr\{X\ge h/2\}
     \le s^{-h/2}(1-p+p s)^h
     =\bigl[2\sqrt{p(1-p)}\bigr]^h.
    \end{equation}
    \Cref{lem:lp-path-certificate} and a union bound prove \cref{eq:lp-finite-bound}. This argument bounds an event depending only on the support, so it permits arbitrary nonzero values at corrupted positions.
    
    For the fixed-weight assertion, put $\rho=t/m$. If $t=0$ the LP has the zero word as its unique optimal marginal vector. Otherwise, for a uniformly chosen $t$-subset of edges and $s\ge1$, expansion of $s^X$ gives
    \begin{equation}\label{eq:lp-hypergeom-mgf}
     \mathbb E s^X = \sum_{j=0}^h\binom hj(s-1)^j\frac{(t)_j}{(m)_j} \leq (1-\rho+\rho s)^h.
    \end{equation}
    Here $(a)_j=a(a-1)\cdots(a-j+1)$, with $(t)_j=0$ for $j>t$, and $(t)_j/(m)_j\le(t/m)^j$. Applying \cref{eq:lp-path-tail} with $\rho$ in place of $p$ proves the same bound because $\rho\le p<1/2$ and $2\sqrt{p(1-p)}$ is increasing on $[0,1/2]$.
\end{proof}

For independent noise at any fixed positive rate in \cref{cor:lp-linial}, the error weight is $p m+o(m)$ with probability tending to one. The fixed-weight version proves the linear-error statement without using this concentration step. Neither version gives correction of all adversarial supports of linear size.

\paragraph{Recovery after a one-coordinate change.}
Changing one coordinate can increase the number of corrupted positions on a simple path by at most one. Hence, if every simple path of $h<g_G$ edges contains fewer than $h/2-1$ corrupted positions, then every word at Hamming distance at most one from the error satisfies the hypotheses of \cref{lem:lp-path-certificate}. Each such word is therefore uniquely LP-recoverable from its own syndrome. With $s=(1-p)/p$, the tail estimate \cref{eq:lp-path-tail} becomes
\begin{equation}
 \Pr\{X\ge h/2-1\}
 \le s^{1-h/2}(1-p+ps)^h
 =\frac{1-p}{p}\bigl[2\sqrt{p(1-p)}\bigr]^h.
\end{equation}
For a uniformly chosen support of size $t\le pm$, use \cref{eq:lp-hypergeom-mgf} with the same $s$ and $t/m\le p$. A union bound proves \cref{eq:lp-robust-finite-bound} in both cases. The bound is uniform over all nonzero values, including replacements introduced by the one-coordinate change.

\subsection{An equivalent LP with polynomial dependence on the degree}
\label{sec:lp-trellis}
The recovery proof is complete. We now establish the stronger size bound stated in \cref{sec:lp-decoder}. A check involving $d_v\ge1$ symbols has exactly $q^{d_v-1}$ satisfying assignments: the first $d_v-1$ symbol values can be chosen freely, and the check uniquely determines the last value because its coefficient is nonzero. After omitting empty checks, the direct LP therefore has
\begin{equation}
    mq+\sum_{v:d_v>0}q^{d_v-1}\le mq+nq^{D-1}
\end{equation}
variables. Its number of constraints and their encoding size are also polynomial in $n$, $m$, and $q$ for fixed $D$, and all scalar coefficients are $0,\pm1$. Thus the direct formulation already admits polynomial-time LP optimization in that setting. The construction below instead has $O(mq^2)$ variables and constraints, hence $O(nDq^2)$ for maximum degree $D$. Its size has no exponential dependence on $D$.

\paragraph{Representing a parity check by a layered graph.}
Fix a syndrome $\mathbf s\in\operatorname{im}H$. For each nonempty check, we construct a separate auxiliary directed graph, called a \emph{trellis}. Its nodes are arranged in successive layers, and its directed edges go from one layer to the next. These nodes are not vertices of $G$: they record running sums as the symbols in one check are assigned values.

Fix a vertex $v$ with $d_v>0$, and order the edges in its parity-check support as $\partial v=\{e_1,\ldots,e_{d_v}\}$. At layer $i=0,\ldots,d_v$, include one node $(i,z)$ for each $z\in\mathbb F_q$. The node records the weighted partial sum $z=\sum_{j=1}^i h_{ve_j}b_{e_j}$. Assigning the next symbol value $b_{e_i}=a$ is represented by the directed edge
\begin{equation}
    (i-1,z)\rightarrow(i,z+h_{ve_i}a), \qquad \text{labeled }a.
\end{equation}
We include this edge for every $i=1,\ldots,d_v$ and $z,a\in\mathbb F_q$. Starting at $(0,0)$, each sequence of symbol values determines exactly one path. The assignment satisfies the check precisely when that path ends at $(d_v,s_v)$. Different partial assignments with the same weighted sum reach the same node; the graph need not store them separately.

\paragraph{A binary example.}
For the check $b_1+b_2+b_3=0$ over $\mathbb F_2$, the state at layer $i$ is the sum of the first $i$ symbols modulo two. Choosing $0$ preserves the sum, whereas choosing $1$ changes it to the other value. The full trellis has eight nodes and twelve directed edges, as shown in \cref{fig:lp-binary-trellis}. The four paths from $(0,0)$ to $(3,0)$ represent exactly the satisfying assignments $000$, $011$, $101$, and $110$. For example, the assignments $000$ and $110$ both reach $(2,0)$ after two steps and use the same final edge. Their earlier choices are recorded by their paths, not by separate copies of the node $(2,0)$.

\begin{figure}[htbp]
\centering
\begin{tikzpicture}[
    x=1cm,y=1cm,>=stealth,
    state/.style={circle,draw,minimum size=12mm,inner sep=1pt,font=\small},
    unused/.style={state,draw=gray,dashed,text=gray},
    active/.style={->,semithick},
    inactive/.style={->,semithick,dashed,draw=gray},
    bit/.style={font=\small,fill=white,inner sep=1.5pt}]
    \node[state] (v00) at (0,0) {$(0,0)$};
    \node[unused] (v01) at (0,-2.1) {$(0,1)$};
    \node[state] (v10) at (3.3,0) {$(1,0)$};
    \node[state] (v11) at (3.3,-2.1) {$(1,1)$};
    \node[state] (v20) at (6.6,0) {$(2,0)$};
    \node[state] (v21) at (6.6,-2.1) {$(2,1)$};
    \node[state] (v30) at (9.9,0) {$(3,0)$};
    \node[unused] (v31) at (9.9,-2.1) {$(3,1)$};
    \foreach \x/\i in {0/0,3.3/1,6.6/2,9.9/3}
        \node[font=\small] at (\x,1.25) {Layer $\i$};
    \node[font=\scriptsize] at (0,0.82) {start};
    \node[font=\scriptsize] at (9.9,0.82) {required end};
    \node[font=\scriptsize,text=gray] at (0,-2.95) {unused initial state};
    \node[font=\scriptsize,text=gray] at (9.9,-2.95) {wrong final sum};
    \draw[active] (v00) -- node[bit,above] {$0$} (v10);
    \draw[active] (v00) -- node[bit,pos=.24,above] {$1$} (v11);
    \draw[inactive] (v01) -- node[bit,pos=.24,below,text=gray] {$1$} (v10);
    \draw[inactive] (v01) -- node[bit,below,text=gray] {$0$} (v11);
    \draw[active] (v10) -- node[bit,above] {$0$} (v20);
    \draw[active] (v10) -- node[bit,pos=.24,above] {$1$} (v21);
    \draw[active] (v11) -- node[bit,pos=.24,below] {$1$} (v20);
    \draw[active] (v11) -- node[bit,below] {$0$} (v21);
    \draw[active] (v20) -- node[bit,above] {$0$} (v30);
    \draw[inactive] (v20) -- node[bit,pos=.24,above,text=gray] {$1$} (v31);
    \draw[active] (v21) -- node[bit,pos=.24,below] {$1$} (v30);
    \draw[inactive] (v21) -- node[bit,below,text=gray] {$0$} (v31);
\end{tikzpicture}

\smallskip
{\small
\begin{tabular}{cc}
\toprule
Assignment $(b_1,b_2,b_3)$ & Running sums $(z_0,z_1,z_2,z_3)$\\
\midrule
$(0,0,0)$ & $(0,0,0,0)$\\
$(0,1,1)$ & $(0,0,1,0)$\\
$(1,0,1)$ & $(0,1,1,0)$\\
$(1,1,0)$ & $(0,1,0,0)$\\
\bottomrule
\end{tabular}}
\caption{Full trellis for $b_1+b_2+b_3=0$ over $\mathbb F_2$. An edge label gives the chosen symbol value, and $(i,z)$ records the running sum after $i$ choices. The table lists all four paths from the required start to the required end. Dashed nodes and edges are included in the full construction but lie on no such path; they carry zero flow in the LP. Crossings of edges are not nodes.}
\label{fig:lp-binary-trellis}
\end{figure}
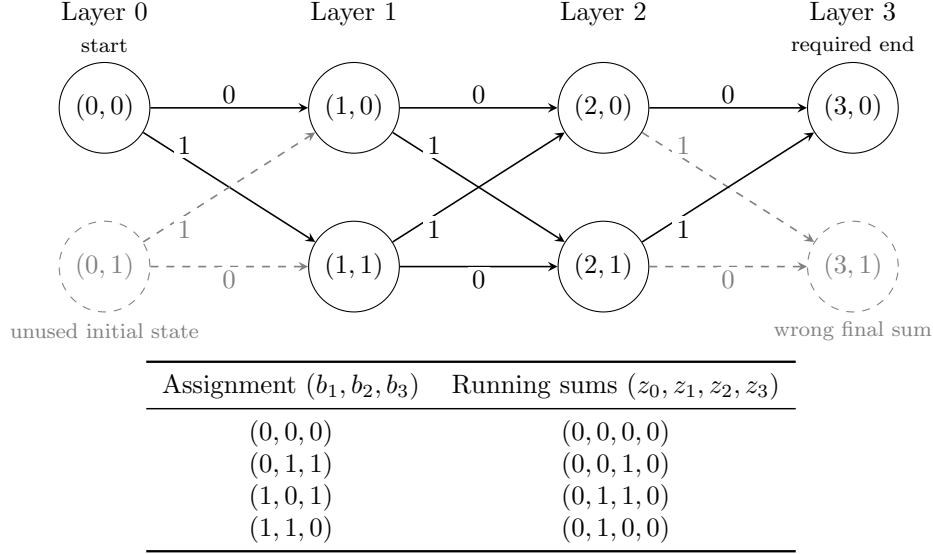

\paragraph{From assignment probabilities to flows.}
A probability distribution over satisfying assignments is a probability distribution over paths from $(0,0)$ to $(d_v,s_v)$. Send each path's probability along its directed edges. The resulting nonnegative edge flows have total supply one at $(0,0)$, total demand one at $(d_v,s_v)$, and flow conservation at all other nodes, with zero supply or demand elsewhere. Conversely, any flow obeying these constraints in this acyclic graph decomposes into a convex combination of such paths, hence into a distribution over satisfying assignments.

Let $f_{v,i,z,a}$ be the flow on the edge labeled $a$ from $(i-1,z)$. The total flow choosing $a$ for symbol $e_i$ must equal its marginal:
\begin{equation}
 x_{e_i,a}=\sum_{z\in\mathbb F_q}f_{v,i,z,a}.
\end{equation}
Use the same variables $x_{e,a}$ in the trellises for both endpoints of $e$, and keep the objective $\mathcal L(\mathbf x)$ unchanged. The two formulations then have exactly the same feasible symbol marginals, and therefore the same optimal marginal vectors. The reformulation neither tightens nor weakens the LP relaxation, so all recovery guarantees proved above apply to it. A vertex with $d_v=0$ imposes only $s_v=0$ and can be omitted because the syndrome is valid.

\paragraph{Size and complexity.}
Each nonempty check uses $d_vq^2$ directed edges and $O(d_vq)$ node constraints. Since $\sum_v d_v=2m\le nD$, including the shared symbol marginals and their consistency constraints gives $O(mq^2)$ variables and constraints in total. Each flow variable occurs in only two node equations and one marginal equation, so the number of nonzero scalar coefficients is also $O(mq^2)$. The field entries $h_{ve}$ determine which nodes a directed edge connects; the real-valued LP constraints themselves have coefficients $0,\pm1$, as does the nonconstant part of the objective. Rational LP optimization thus takes time polynomial in $n$, $q$, and $D$. A fixed lexicographic rule on the symbol marginals can be implemented by a polynomial number of further LP optimizations, making the decoder deterministic in either formulation. This removes the exponential degree dependence in the implementation; it does not change the parameter regimes of the recovery results.

\end{appendices}
\addtocontents{toc}{\protect\setcounter{tocdepth}{2}}

\begingroup
\emergencystretch=1em
\printbibliography

@article{apte2026provable,
  title={Provable Classical and Quantum Local Algorithms for Max-$ k $-Cut and Quantum Advantage at Moderate Girth},
  author={Apte, Anuj and Khan, Abid and Farhi, Edward and Marwaha, Kunal and Boulebnane, Sami and Shaydulin, Ruslan},
  journal={arXiv preprint arXiv:2609.39042},
  year={2026}
}

@article{JSW25,
	title = {Optimization by decoded quantum interferometry},
	author = {Stephen P. Jordan and Noah Shutty and Mary Wootters and Adam Zalcman and Alexander Schmidhuber and Robbie King and Sergei V. Isakov and Tanuj Khattar and Ryan Babbush},
  date = {2025-10-01},
	doi = {10.1038/s41586-025-09527-5},
  issn = {1476-4687},
	journal = {Nature},
	number = {8086},
	pages = {831--836},
	url = {https://doi.org/10.1038/s41586-025-09527-5},
	volume = {646},
	year = {2025},
    note={arXiv:2408.08292}
}

@article{kramer2026approximability,
  title={Approximability limits for bounded-degree max-LINSAT and implications for decoded quantum interferometry},
  author={Kramer, Maximilian J and Schubert, Carsten and Eisert, Jens},
  journal={arXiv preprint arXiv:2606.13570},
  year={2026}
}

@article{marwaha2025complexity,
  title={On the complexity of decoded quantum interferometry},
  author={Marwaha, Kunal and Fefferman, Bill and Gheorghiu, Alexandru and Havlicek, Vojtech},
  journal={arXiv preprint arXiv:2509.14443},
  year={2025}
}

@article{P25,
  title={No quantum advantage in decoded quantum interferometry for maxcut},
  author={Ojas Parekh},
  journal={arXiv:2509.19966},
  year={2025}
}

@article{EJ73,
  title={Matching, {E}uler tours and the {C}hinese postman},
  author={Jack Edmonds and Ellis L. Johnson},
  journal={Mathematical programming},
  volume={5},
  pages={88--124},
  year={1973},
  publisher={Springer}
}

@article{DZ97,
  title={On the error-correcting capabilities of cycle codes of graphs},
  author={Laurent Decreusefond and Gilles Z{\'e}mor},
  journal={Combinatorics, Probability and Computing},
  volume={6},
  number={1},
  pages={27--38},
  year={1997},
  publisher={Cambridge University Press}
}

@book{GJ79,
	address = {New York},
	edition = {27. print},
	series = {A series of books in the mathematical sciences},
	title = {Computers and intractability: a guide to the theory of {NP}-completeness},
	shorttitle = {Computers and intractability},
	language = {eng},
	publisher = {Freeman},
	author = {Michael R. Garey and David S. Johnson},
	year = {2009},
}

@article{mckay2004,
	title = {Short cycles in random regular graphs},
	volume = {11},
	issn = {1077-8926},
	url = {https://www.combinatorics.org/ojs/index.php/eljc/article/view/v11i1r66},
	doi = {10.37236/1819},
	number = {1},
	journal = {The Electronic Journal of Combinatorics},
	author = {Brendan D. McKay and Nicholas C. Wormald and Beata Wysocka},
    date = {2004-09},
	pages = {R66}
}

@misc{EH12,
  author={Guy Even and Nissim Halabi},
  title={Local-Optimality Guarantees for Optimal Decoding Based on Paths},
  year={2012},
  eprint={1203.1854},
  eprinttype={arXiv},
  url={https://arxiv.org/abs/1203.1854}
}

@article{FSBG09,
  author={Mark F. Flanagan and Vitaly Skachek and Eimear Byrne and Marcus Greferath},
  title={Linear-Programming Decoding of Nonbinary Linear Codes},
  journal={IEEE Transactions on Information Theory},
  volume={55},
  number={9},
  pages={4134--4154},
  year={2009},
  doi={10.1109/TIT.2009.2025571},
  eprint={0804.4384},
  eprinttype={arXiv}
}

@misc{LS,
  author={Nathan Linial and Michael Simkin},
  title={A Randomized Construction of High Girth Regular Graphs},
  year={2020},
  eprint={1911.09640},
  eprinttype={arXiv},
  note={Version 3},
  url={https://arxiv.org/abs/1911.09640}
}

@Article{Frieze1997,
author={A. Frieze and M. Jerrum},
title={Improved approximation algorithms for MAXk-CUT and MAX BISECTION},
journal={Algorithmica},
date={1997-05},
volume={18},
number={1},
pages={67-81},
issn={1432-0541},
doi={10.1007/BF02523688},
url={https://doi.org/10.1007/BF02523688}
}

@article{gu2025algebraic,
  title={Algebraic Geometry Codes and Decoded Quantum Interferometry},
  author={Andi Gu and Stephen P. Jordan},
  journal={arXiv:2510.06603},
  year={2025}
}

@article{Johnson15,
  author={Tobias Johnson},
  title={Exchangeable Pairs, Switchings, and Random Regular Graphs},
  journal={The Electronic Journal of Combinatorics},
  volume={22},
  number={1},
  pages={P1.33},
  year={2015},
  doi={10.37236/4659},
  note={arXiv:1112.0704}
}

@article{NVZ16,
  author={Peter Nelson and Stefan H. M. van Zwam},
  title={The Maximum-Likelihood Decoding Threshold for Cycle Codes of Graphs},
  journal={IEEE Transactions on Information Theory},
  volume={62},
  number={10},
  pages={5316--5322},
  year={2016},
  note={arXiv:1504.05225}
}

@book{Biggs1993AlgebraicGraphTheory,
  title={Algebraic graph theory},
  author={Biggs, N.L.},
  year=1993,
  series={Algebraic Graph Theory},
  url={https://books.google.com/books?id=WO6inQEACAAJ},
  publisher={Cambridge University Press}
}

@book{Bollobas2001RandomGraphs,
  title={Random Graphs},
  author={Bollob{\'a}s, B.},
  isbn={9780521797221},
  lccn={00068952},
  series={Cambridge Studies in Advanced Mathematics},
  url={https://books.google.com/books?id=o9WecWgilzYC},
  year={2001},
  publisher={Cambridge University Press}
}

@inproceedings{CLZ22,
  title={Quantum algorithms for variants of average-case lattice problems via filtering},
  author={Yilei Chen and Qipeng Liu and Mark Zhandry},
  booktitle={Annual international conference on the theory and applications of cryptographic techniques},
  pages={372--401},
  year={2022},
  organization={Springer}
}

@article{KOW25,
  title={No exponential quantum speedup for {$\mathrm{SIS}^\infty$} anymore},
  author={Robin Kothari and Ryan O'Donnell and Kewen Wu},
  journal={arXiv:2510.07515},
  year={2025}
}

@inproceedings{TPM22,
  author={Jessica K. Thompson and Ojas Parekh and Kunal Marwaha},
  title={An Explicit Vector Algorithm for High-Girth {MaxCut}},
  booktitle={2022 Symposium on Simplicity in Algorithms (SOSA)},
  year={2022},
  publisher={SIAM},
  doi={10.1137/1.9781611977066.17},
  eprint={2108.12477},
  eprinttype={arXiv}
}

@techreport{KB78,
  author={R. Kaas and J. M. Buhrman},
  title={A Note on the Median of the Binomial Distribution},
  institution={Stichting Mathematisch Centrum},
  number={SW 59/78},
  year={1978},
  url={https://ir.cwi.nl/pub/8056}
}

@phdthesis{Como08,
  author = {Giacomo Como},
  title = {Ensembles of codes over {Abelian} groups},
  school = {Politecnico di Torino},
  year = {2008},
  month = mar,
  url = {https://archive.control.lth.se/media/Staff/GiacomoComo/ComoThesis.pdf}
}
\endgroup

\end{document}